\documentclass[a4paper,twocolumn,11pt,accepted=2023-11-17]{quantumarticle}
\pdfoutput=1
\usepackage[utf8]{inputenc}
\usepackage[english]{babel}
\usepackage[T1]{fontenc}

\usepackage{amsmath}
\usepackage{mathtools}
\usepackage{amssymb}
\usepackage{amsthm}
\usepackage{thm-restate}
\usepackage{bbold}
\usepackage{xcolor}
\usepackage{physics}
\usepackage{hyperref}
\usepackage{bm}
\usepackage{ytableau}
\theoremstyle{plain}
\newtheorem{theorem}{Theorem}

\newtheorem{lemma}[theorem]{Lemma}
\newtheorem{corollary}[theorem]{Corollary}
\theoremstyle{definition}
\newtheorem{definition}[theorem]{Definition}

\theoremstyle{remark}
\newtheorem{remark}[theorem]{Remark}

\newcommand{\changevtwo}[1]{#1}
\newcommand{\changevthree}[1]{#1}
\newcommand{\changevfour}[1]{#1}

\usepackage{tikz}
\usetikzlibrary{decorations,decorations.pathmorphing,decorations.markings,calc,snakes,positioning,fixedpointarithmetic,shapes,shapes.geometric,patterns}

\usepackage{etoolbox}
\AtBeginEnvironment{tikzpicture}{\catcode`$=3 } 

\tikzset{alignmid/.style={baseline={([yshift=-.5ex]current bounding box.center)}}} 
\tikzset{every picture/.append style=alignmid}

\tikzset{
bottomzigzag/.style={postaction={draw,decorate, decoration={zigzag,amplitude=1pt,segment length=3pt,raise=1pt}}},
zigzag/.style={draw,decorate, decoration={zigzag,amplitude=1pt,segment length=3pt}},
rc/.style=rounded corners,
}

\tikzset{
    -|/.style={to path={-| (\tikztotarget)}},
    |-/.style={to path={|- (\tikztotarget)}},
}

\usepackage{ifthen}
\usepackage{xparse,pgffor}

\tikzset{
mark/.code={
\tikzset{postaction={/network/mark/.cd,#1,/tikz/.cd,decorate},decoration={name=markings,mark=at position \netmarkpos with{
\begin{scope}[netmarktrafo]
\netmarkcode
\end{scope}
}}}
\def\netmarkpos{0.5}
},
}

\def\netmarkpos{0.5}
\def\netmarkposoff{0cm}
\def\netmarkcode{}

\tikzset{
netmarktrafo/.style={},
netmarkstyle/.style={solid,semithick,sharp corners},
}

\pgfkeys{
/network/mark/.cd,
sty/.code={
\tikzset{netmarkstyle/.style={#1}}
},
asty/.code={
\tikzset{netmarkstyle/.append style={#1}}
},
f/.style={asty=fill},
p/.code={ 
\def\netmarkpos{#1}
},
poff/.code={ 
\def\netmarkposoff{#1cm}
},
e/.code={ 
\def\netmarkpos{\pgfdecoratedpathlength-0.005cm-\netmarkposoff}

\tikzset{netmarktrafo/.append style={shift={(-\netmarkwidth,0)}}}
},
s/.code={ 
\def\netmarkpos{0.005cm+\netmarkposoff}

\tikzset{netmarktrafo/.append style={shift={(\netmarkwidth,0)},xscale=-1,yscale=-1}}
},
a/.code={ 
\def\netmarkpos{\pgfdecoratedpathlength-0.005cm}

\tikzset{netmarktrafo/.append style={xscale=-1,shift={(-\netmarkwidth,0)}}}
},
b/.code={ 
\def\netmarkpos{0.005cm}

\tikzset{netmarktrafo/.append style={xscale=-1,shift={(\netmarkwidth,0),yscale=-1}}}
},
-/.code={ 
\tikzset{netmarktrafo/.append style={xscale=-1}}
},
r/.code={ 
\tikzset{netmarktrafo/.append style={yscale=-1}}
},
sideoff/.code={ 
\tikzset{netmarktrafo/.append style={shift={(0,#1)}}}
},
slab/.code={ 
\def\netmarkwidth{0}
\def\netmarkcode{
\node[inner sep=0.04cm,netmarkstyle,draw=none] (mylabelwidthtest) at (0,0){\phantom{#1}};
\path let \p1=(mylabelwidthtest.north east), \p2=(mylabelwidthtest.south east), \n1 = {max(abs(\y1),abs(\y2))} in node[inner sep=0.04cm,netmarkstyle] at (0,\n1) {#1};
}
},
lab/.code={ 
\def\netmarkwidth{0}
\def\netmarkcode{
\node[inner sep=0.04cm,anchor=\netmarkanchor] (mylabelwidthtest) at (0,0) {\phantom{#1}};
\draw[white] (mylabelwidthtest.\pgfdecoratedangle)--(mylabelwidthtest.\pgfdecoratedangle+180);
\node[inner sep=0.04cm,anchor=\netmarkanchor,netmarkstyle] at (0,0) {#1};
}
},
arr/.code={ 
\def\netmarkwidth{0.04}
\def\netmarkcode{\draw[netmarkstyle] (-0.04,0.08)--(0.04,0)--(-0.04,-0.08);}
},
rarr/.code={ 
\def\netmarkwidth{0.04}
\def\netmarkcode{\draw[netmarkstyle] (-0.04,-0.08)arc(90-180:90:0.08);}
},
dot/.code={ 
\def\netmarkwidth{0.08}
\def\netmarkcode{\draw[netmarkstyle] (0,0)circle(0.08);}
},
flag/.code={ 
\def\netmarkwidth{0.06}
\def\netmarkcode{\draw[netmarkstyle] (-0.06,0)--(0,0.09)--(0.06,0)--cycle;}
},
darr/.code={ 
\def\netmarkwidth{0.08}
\def\netmarkcode{\draw[netmarkstyle] (-0.04,0)--(0.04,0)--(-0.04,0.08)--cycle;}
},
hcirc/.code={ 
\def\netmarkwidth{0.1}
\def\netmarkcode{\draw[netmarkstyle] (-0.1,0) arc (180:0:0.1);}
},
bar/.code={ 
\def\netmarkwidth{0.05}
\def\netmarkcode{
\draw[netmarkstyle] (0,-0.08cm-0.5*\pgflinewidth)--(0,0.08cm+0.5*\pgflinewidth);
}
},
hbar/.code={ 
\def\netmarkwidth{0.05}
\def\netmarkcode{
\draw[netmarkstyle] (0, 0.5*\pgflinewidth)--++(0,0.12);
}
},
mill/.code={
\def\netmarkwidth{0.16}
\def\netmarkcode{
\draw[netmarkstyle] (0,-0.5*\pgflinewidth)--++(-0.08,-0.08)--++(0,0.08);
\draw[netmarkstyle] (0,0.5*\pgflinewidth)--++(0.08,0.08)--++(0,-0.08);
}
},
three dots/.code={
\def\netmarkwidth{0.2}
\def\netmarkcode{
\fill (-0.12,0) circle (0.5*0.05) (0,0) circle (0.5*0.05) (0.12,0) circle (0.5*0.05);
}
},
}

\tikzset{wid/.style={minimum width=#1cm}}
\tikzset{hei/.style={minimum height=#1cm}}

\tikzset{sx/.style={xshift=#1cm}}
\tikzset{sy/.style={yshift=#1cm}}

\tikzset{box/.style={draw,rectangle}}
\tikzset{fbox/.style={draw,rectangle, line width=1.1}}
\tikzset{roundbox/.style={draw,rectangle,rounded corners}}
\tikzset{froundbox/.style={draw,rectangle, rounded corners, line width=1.1}}
\tikzset{rounddiamond/.style={draw,diamond,rounded corners}}
\tikzset{dot/.style={draw, shape=circle, fill=black, scale=0.5}}

\tikzset{
netbox/.code={
\node[draw,netbdstyle] (\atomname) at (0,0) {#1};
\coordinate (\atomname-r) at (\atomname.east);
\coordinate (\atomname-l) at (\atomname.west);
\coordinate (\atomname-t) at (\atomname.north);
\coordinate (\atomname-b) at (\atomname.south);
\coordinate (\atomname-tr) at (\atomname.north east);
\coordinate (\atomname-br) at (\atomname.south east);
\coordinate (\atomname-tl) at (\atomname.north west);
\coordinate (\atomname-bl) at (\atomname.south west);
},
}

\tikzset{bdlw/.code={\tikzset{mybdstyle/.style={draw, line width=#1}}}}
\tikzset{bdcol/.code={\tikzset{mybdstyle/.append style={#1}}}}

\newcommand\setelements[1]{
\pgfkeys{/network/atom/.cd,#1}
}

\newcommand\atoms[2]{
\foreach \name/\keys in {#2}{
\expandafter\atom\expandafter{\keys,#1}{\name}
}
}

\newcommand\atom[2]{
\def\atomname{#2}
\tikzset{
nettrafo/.style={},
netatompos/.style={},
netdeco/.style={},
netpostdeco/.style={},
}

\pgfkeys{/network/atom/.cd,#1}

\begin{scope}[netatompos] 
\begin{scope}[nettrafo] 
\netshapecoords 
\fill[netbackstyle] \netshapepath;
\clip \netshapepath;
\tikzset{netdeco}
\draw[netbdstyle] \netshapepath;
\end{scope}
\tikzset{netpostdeco} 
\end{scope}

}

\pgfkeys{
/network/atom/.cd,
square/.code={

\def\netshapepath{(-\tempsize,-\tempsize)rectangle (\tempsize,\tempsize)}
\def\netshapecoords{
\node[rectangle,wid=2*\tempsize,hei=2*\tempsize,inner sep=0,transform shape](\atomname)at(0,0){};
\coordinate(\atomname-c) at (0,0);
\coordinate(\atomname-r) at (\tempsize,0);
\coordinate(\atomname-l) at (-\tempsize,0);
\coordinate(\atomname-t) at (0,\tempsize);
\coordinate(\atomname-b) at (0,-\tempsize);
\coordinate(\atomname-br) at (\tempsize,-\tempsize);
\coordinate(\atomname-tr) at (\tempsize,\tempsize);
\coordinate(\atomname-bl) at (-\tempsize,-\tempsize);
\coordinate(\atomname-tl) at (-\tempsize,\tempsize);
}},
circ/.code={

\def\netshapepath{(0,0)circle(\tempsize)}
\def\netshapecoords{
\node[circle,wid=2*\tempsize,hei=2*\tempsize,inner sep=0,transform shape](\atomname)at(0,0){};
\coordinate(\atomname-c) at (0,0);
\coordinate(\atomname-r) at (\tempsize,0);
\coordinate(\atomname-l) at (-\tempsize,0);
\coordinate(\atomname-t) at (0,\tempsize);
\coordinate(\atomname-b) at (0,-\tempsize);
}},
triang/.code={

\def\netshapepath{(-30:\tempsize)--(90:\tempsize)--(-150:\tempsize)--cycle}
\def\netshapecoords{
\node[regular polygon,regular polygon sides=3,wid=2*\tempsize,inner sep=0,transform shape](\atomname)at(0,0){};
\coordinate(\atomname-c) at (0,0);
\coordinate(\atomname-cr) at (-30:\tempsize);
\coordinate(\atomname-cl) at (-150:\tempsize);
\coordinate(\atomname-ct) at (90:\tempsize);
\coordinate(\atomname-mb) at (-90:0.5*\tempsize);
\coordinate(\atomname-mr) at (30:0.5*\tempsize);
\coordinate(\atomname-ml) at (150:0.5*\tempsize);
}},
diamond/.code={

\def\netshapepath{(0,-\tempsize)--(\tempsize,0)--(0,\tempsize)--(-\tempsize,0)--cycle}
\def\netshapecoords{
\node[rotate=45,rectangle,wid=sqrt(2)*\tempsize,hei=sqrt(2)*\tempsize,inner sep=0,transform shape](\atomname)at(0,0){};
\coordinate(\atomname-c) at (0,0);
\coordinate(\atomname-r) at (\tempsize,0);
\coordinate(\atomname-l) at (-\tempsize,0);
\coordinate(\atomname-t) at (0,\tempsize);
\coordinate(\atomname-b) at (0,-\tempsize);
}},
pentagon/.code={

\def\netshapepath{(-126:\tempsize)--(-54:\tempsize)--(18:\tempsize)--(90:\tempsize)--(162:\tempsize)--cycle}
\def\netshapecoords{
\node[regular polygon,regular polygon sides=5,wid=2*\tempsize,inner sep=0,transform shape](\atomname)at(0,0){};
\coordinate(\atomname-c) at (0,0);
\coordinate (\atomname-mb)at(-90:{\tempsize*cos(36)});
\coordinate (\atomname-mbr)at(-18:{\tempsize*cos(36)});
\coordinate (\atomname-mtr)at(54:{\tempsize*cos(36)});
\coordinate (\atomname-mtl)at(126:{\tempsize*cos(36)});
\coordinate (\atomname-mbl)at(-162:{\tempsize*cos(36)});
\coordinate (\atomname-cbr)at(-54:\tempsize);
\coordinate (\atomname-cr)at(18:\tempsize);
\coordinate (\atomname-ct)at(90:\tempsize);
\coordinate (\atomname-cl)at(162:\tempsize);
\coordinate (\atomname-cbl)at(-126:\tempsize);
}},
halfcirc/.code={

\def\netshapepath{(\tempsize,0)arc(0:180:\tempsize)--++(0,-0.04)-|cycle}
\def\netshapecoords{
\node[circle,wid=2*\tempsize,hei=2*\tempsize,inner sep=0,transform shape](\atomname)at(0,0){};
\coordinate(\atomname-c) at (0,0);
\coordinate(\atomname-r) at (\tempsize,0);
\coordinate(\atomname-l) at (-\tempsize,0);
\coordinate(\atomname-t) at (0,\tempsize);
\coordinate(\atomname-b) at (0,0);
}},
void/.code={
\def\netshapepath{}
\def\netshapecoords{
\coordinate(\atomname) at (0,0);
\coordinate(\atomname-c) at (0,0);
}},
labbox/.code={
\def\netshapepath{(0,0)}
\def\netshapecoords{}
\tikzset{netpostdeco/.append style={netbox=#1}}
},
}

\pgfkeys{
/network/atom/.cd,
p/.code={\tikzset{netatompos/.append style={shift={(#1)}}}},
xscale/.code={\tikzset{nettrafo/.append style={xscale=#1}}},
yscale/.code={\tikzset{nettrafo/.append style={yscale=#1}}},
scale/.code={\tikzset{nettrafo/.append style={scale=#1}}},
rot/.code={\tikzset{nettrafo/.append style={rotate=#1}}},
hflip/.style={xscale=-1},
vflip/.style={yscale=-1},
big/.style={scale=3/2},
small/.style={scale=2/3},
huge/.style={scale=9/4},
tiny/.style={scale=4/9},
flat/.style={yscale=2/3},
fflat/.style={yscale=4/9},
}

\tikzset{
netbdstyle/.style={line width=0.15em}, 
netdecstyle/.style={},
netpostdecstyle/.style={},
netbackstyle/.style={white},
}
\pgfkeys{
/network/atom/.cd,
bdstyle/.code={\tikzset{netbdstyle/.style={#1}}},
bdastyle/.code={\tikzset{netbdstyle/.append style={#1}}},
decstyle/.code={\tikzset{netdecstyle/.style={#1}}},
postdecstyle/.code={\tikzset{netpostdecstyle/.style={#1}}},
backstyle/.code={\tikzset{netbackstyle/.style={#1}}},
whiteblack/.style={decstyle=white,backstyle=black},
nobd/.code={bdstyle={line width=0}},
style/.style={bdastyle=#1, decstyle=#1, postdecstyle=#1},
}

\tikzset{
netbscope/.code={\begin{scope}[#1]},
netescope/.code={\end{scope}},
}

\pgfkeys{
/network/atom/.cd,
bscope/.code={\tikzset{netdeco/.append style={netbscope={#1}}}},
escope/.code={\tikzset{netdeco/.append style={netescope}}},
dec/.style 2 args={bscope={#1},#2,escope}, 
vbar/.style={dec={rotate=90}{hbar}},
dcross/.style={dec={rotate=45}{hbar},dec={rotate=-45}{hbar}},
cross/.style={hbar,vbar},
sect6/.style={sect=60},
sect8/.style={sect=45},
sect10/.style={sect=36},
}

\def\regdec#1{\pgfkeys{/network/atom/.cd,#1/.code={\tikzset{netdeco/.append style={net#1}}}}}
\regdec{all}\regdec{rhalf}\regdec{rquart}\regdec{brquart}\regdec{dot}\regdec{spiral}\regdec{swirl}\regdec{hstripe}\regdec{hbar}\regdec{rrey}
\pgfkeys{/network/atom/.cd,sect/.code={\tikzset{netdeco/.append style={netsect=#1}}}}

\tikzset{
netall/.code={\fill[netdecstyle] (-0.3,-0.3)rectangle (0.3,0.3);}, 
netrhalf/.code={\fill[netdecstyle] (0,-0.3)rectangle (0.3,0.3);}, 
netrquart/.code={\fill[netdecstyle] (0.075,-0.3)rectangle (0.3,0.3);}, 
netbrquart/.code={\fill[netdecstyle] (0,0)rectangle (0.3,-0.3);}, 
netsect/.code={\fill[netdecstyle] (0,0)--(0,-0.3)arc(-90:-90+#1:0.3)--cycle;}, 
netdot/.code={\fill[netdecstyle] (0,0)circle(0.07);}, 
netspiral/.code={\draw[netdecstyle] plot [variable=\t,domain=0:4] ({0.075*\t*cos(pi*(\t-0.5) r)},{0.075*\t*sin(pi*(\t-0.5) r)});}, 
netswirl/.code={\fill[netdecstyle] plot [variable=\t,domain=0:2] ({0.15*\t*cos(pi*(\t-0.5) r)},{0.15*\t*sin(pi*(\t-0.5) r)}) arc(-90:-450:0.3)--cycle;}, 
nethstripe/.code={\fill[netdecstyle] (-0.3,-0.05)rectangle(0.3,0.05);}, 
nethbar/.code={\draw[netdecstyle] (-0.3,0)--(0.3,0);}, 
netrrey/.code={\draw[netdecstyle] (0,0)--(0.3,0);} 
}

\pgfkeys{
/network/atom/.cd,
postbscope/.code={\tikzset{netpostdeco/.append style={netbscope={#1}}}},
postescope/.code={\tikzset{netpostdeco/.append style={netescope}}},
postdec/.style 2 args={postbscope={#1},#2,postescope}, 
arc/.code={\tikzset{netpostdeco/.append style={netarc=#1}}},
lab/.code={\tikzset{netpostdeco/.append style={netlab={#1}}}},
shadecirc/.code={\tikzset{netpostdeco/.append style=netshadecirc}},
shaderect/.code={\tikzset{netpostdeco/.append style={netshaderect={#1}}}},
debug/.code={\tikzset{netpostdeco/.append style=netdebug}},
markline/.code 2 args={\tikzset{netpostdeco/.append style={netmarkline={#1}{#2}}}},
postcirc/.code={\tikzset{netpostdeco/.append style=netpostcirc}},
}

\tikzset{
netlab/.code={
\pgfkeys{/network/atom/lab/.cd,#1}
\node[netpostdecstyle] at (\ifdefined\netlabpos\netlabpos\else\netlabang:\netlabdist\fi) {\netlabwrap{\netlabtext}};
},
netarc/.code args={#1:#2:#3}{
\draw[netpostdecstyle] (#1:#3) arc (#1:#2:#3);
},
netshadecirc/.code= {
\fill[opacity=0.4,netpostdecstyle] (0,0)circle(0.4);
},
netpostcirc/.code= {
\draw[netpostdecstyle] (0,0)circle(0.15);
},
netshaderect/.code= {
\fill[rc,opacity=0.4,netpostdecstyle] ($-1*(#1)$) rectangle (#1);
},
netdebug/.code= {
\node[red] at (0,0){\atomname};
},
netmarkline/.code 2 args= {
\draw (\atomname)edge[mark={#2}]++(#1);
},
}

\def\netlabwrap#1{#1}
\pgfkeys{
/network/atom/lab/.cd,
t/.code={\def\netlabtext{#1}}, 
p/.code={\def\netlabpos{#1}}, 
ang/.code={\def\netlabang{#1}},
dist/.code={\def\netlabdist{#1}},
wrap/.code={\def\netlabwrap##1{#1}}, 
}

\usetikzlibrary{decorations,decorations.pathmorphing}

\tikzset{
ind/.style={mark={lab=#1,a}},
nqubit/.style={line width=3},
triple/.style={preaction={preaction={draw,line width=6},draw,line width=5.2,white}},
irrep/.style={line width=1.5},
qind/.style={thick, densely dotted},
pind/.style={very thick,postaction={sharp corners,draw,decorate,decoration={ticks,segment length=0.08cm,amplitude=0.05cm,pre=curveto,pre length=0.07cm,post=curveto,post length=0.03cm}}},
fusion/.style={draw=none,postaction={sharp corners,draw,decorate,decoration={zigzag,amplitude=0.04cm,segment length=0.08cm}}},
mps/.style={},
smallmultip/.style={postaction={sharp corners,draw,decorate,decoration={ticks,segment length=0.08cm,amplitude=0.03cm,pre=curveto,pre length=0.07cm,post=curveto,post length=0.03cm}}},
}

\setelements{
delta/.style={circ,tiny,all},
}

\def\gse{\operatorname{GSE}}

\begin{document}
\title{Efficient classical algorithms for simulating symmetric quantum systems}
\author{Eric R.\ Anschuetz}
\email{eans@mit.edu}
\affiliation{MIT Center for Theoretical Physics, 77 Massachusetts Avenue, Cambridge, MA 02139, USA}
\author{Andreas Bauer}
\email{andibauer@zedat.fu-berlin.de}
\affiliation{Dahlem Centre for Complex Quantum Systems, Freie Universit{\"a}t Berlin, Arnimallee 14, 14195 Berlin, Germany}
\author{Bobak T.\ Kiani}
\email{bkiani@mit.edu}
\affiliation{MIT Department of Electrical Engineering and Computer Science, 77 Massachusetts Avenue, Cambridge, MA 02139, USA}
\author{Seth Lloyd}
\email{slloyd@mit.edu}
\affiliation{MIT Department of Mechanical Engineering, 77 Massachusetts Avenue, Cambridge, MA 02139, USA }
\affiliation{Turing Inc., Cambridge, MA 02139, USA}

\begin{abstract}
    In light of recently proposed quantum algorithms that incorporate symmetries in the hope of quantum advantage, we show that with symmetries that are restrictive enough, classical algorithms can efficiently emulate their quantum counterparts given certain classical descriptions of the input.
    Specifically, we give classical algorithms that calculate ground states and time-evolved expectation values for permutation-invariant Hamiltonians specified in the symmetrized Pauli basis with runtimes polynomial in the system size.
    We use tensor-network methods to transform symmetry-equivariant operators to the block-diagonal Schur basis that is of polynomial size, and then perform exact matrix multiplication or diagonalization in this basis. 
    These methods are adaptable to a wide range of input and output states including those prescribed in the Schur basis, as matrix product states, or as arbitrary quantum states when given the power to apply low depth circuits and single qubit measurements.
\end{abstract}

\maketitle

\section{Introduction}

In the physical sciences, symmetries are useful for simplifying difficult computational tasks by reducing the effective degrees of freedom of the problem. This general principle has been used to find exact solutions to many problems, such as integrable systems~\cite{bethe1931theorie}, topological fixed-point models~\cite{Levin2004}, or conformal field theories~\cite{Belavin1984}. There has been a hope that similar symmetries may enable the efficiency of quantum algorithms for simulating or finding the ground state of a symmetric Hamiltonian. Indeed, it is known that there exist theoretical guarantees for quantum algorithms for finding the ground \changevthree{state~\cite{schatzki2022theoretical} and fast-forwarding quantum dynamics~\cite{Gu2021fastforwarding} of Hamiltonians which commute under the action of the symmetric group $\mathrm{S}_n$ on qubits}. It has also numerically been shown that quantum algorithms are capable of finding the ground state of certain integrable systems~\cite{wiersema2020exploring,anschuetz2022critical} even when the symmetry is not explicitly given to the quantum algorithm \textit{a priori}. \changevtwo{Furthermore, prior work used Lie algebraic methods to efficiently classically simulate operators restricted to a Lie algebra whose dimension is polynomially large (independent of the potentially exponentially large Hilbert space dimension) \cite{somma2006efficient,zeier2011symmetry}.} Quantum machine learning models that are symmetry equivariant are also believed to be more efficiently trainable than their general counterparts~\cite{symmetrywu,anschuetzkiani2022,castelazo2022quantum,equivariant,PRXQuantum.3.030341,ragone2022representation}. These quantum models are partly inspired by classical neural network models that have enjoyed much recent success \cite{bronstein2017geometric,wu2020comprehensive,cohen2016group}. However, restricting quantum algorithms to problems obeying many symmetries potentially allows for efficient classical algorithms which also take advantage of these same symmetries. This raises the natural question: are there efficient classical algorithms capable of performing these tasks?

This is what we investigate here. Intuitively, we show that problems constrained by large symmetry groups yield efficient classical algorithms for computing many properties of interest, as illustrated in Figure~\ref{fig:illustrative_stuff}(a). We first broadly discuss very general classical algorithms for finding the ground state and energy of Hamiltonians constrained by many symmetries. We also consider the problem of simulating dynamics under symmetric Hamiltonians. We then specialize to the case of systems invariant under permutations of its qu\changevthree{b}its. Finally, we dequantize an algorithm for performing binary classification problems using permutation-invariant systems on qubits.

\begin{figure*}[t!]
    \begin{center}
        \includegraphics[]{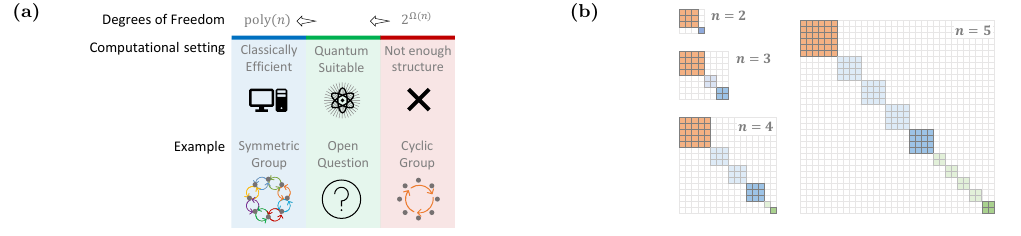}
        \caption{(a) Small groups of symmetry leave too large of an effective dimension for the problem to be tractable via quantum computation. On the contrary, very restrictive symmetries render a problem classically tractable. Between these two regions lies an area of promise where quantum computers may offer an advantage. (b) The Schur--Weyl decomposition shows that only a smaller representative subspace (indicated by darker colors) of the larger Hilbert space needs to be considered for permutation invariant operations. The size of this subspace grows as $\operatorname{O}(n^3)$ for $n$ qubits.\label{fig:illustrative_stuff}}
    \end{center}
\end{figure*}

\section{Motivation and setting}

Our algorithms are motivated by the fact that symmetries significantly reduce the number of degrees of freedom for a given problem. For example, consider the classical setting of Boolean functions which are invariant under arbitrary permutations of the bits. Such functions are defined up to the orbits of the Boolean cube with respect to permutations of the bits. For a Boolean function on $n$ bits, there are $n+1$ orbits indexed by the Hamming weight of the bitstrings. Therefore, any problem over symmetric Boolean functions need only consider a given element of each of the $n+1$ orbits to cover all possible degrees of freedom. As we will later show, the symmetric group acting over $n$ qubits similarly reduces systems to $\operatorname{O}(n^3)$ degrees of freedom. By considering the algebra of the symmetric group on the symmetric subspace of linear operators, we will show that all these degrees of freedom can be manipulated solely through classical computation.

Before proceeding, we need to introduce important functions and definitions that will be used in this setting. We first formalize the notion of symmetry by speaking of \emph{invariant operators}, defined in the following way:
\begin{definition}[Invariant operator]
Given a compact group $G$ with unitary representation $R:G \to U(N)$, a linear operator $H:\mathbb{C}^{N} \to \mathbb{C}^N$ is invariant under $R\left(G\right)$ if
\begin{equation}
    R(g) H R(g)^\dagger = H, \; \; \; \forall g \in G.
\end{equation}
Note that any invariant operator is also an \emph{equivariant operator}~\cite{equivariant} in the sense that it commutes with the representation of the group.
\label{def:inv_op}
\end{definition}

Any operator can be projected onto the symmetric subspace induced by $R\left(G\right)$ using the twirling superoperator $\operatorname{Re}_R$ (more commonly known as the Reynold's operator in invariant theory) \cite{olver1999classical,sturmfels2008algorithms}, which maps any operator onto the set of equivariant operators:
\begin{equation}
\label{eq:reynolds_definition}
    \operatorname{Re}_R(M) = \frac{1}{|G|} \sum_{g \in G} R(g) M R(g)^\dagger.
\end{equation}

Invariant subspaces of a larger Hilbert space can be identified by performing an isotypic decomposition of the representation of a group. As an example, in the case of systems invariant under permutations of the qudits, the Schur decomposition maps the computational basis into blocks of invariant subspaces. We graphically visualize this phenomenon in Figure~\ref{fig:illustrative_stuff}(b) and provide further details in Appendix~\ref{app:schur_basis}.

Throughout this study, runtime complexities are denoted as a function of the matrix multiplication exponent $\omega$. The best known upper bound is currently $\omega=2.37188$ \cite{duan2022faster}, which implies asymptotic runtimes of $\operatorname{O}(n^{\omega + \alpha})$ for any $\alpha>0$ for stably performing common linear algebraic routines including eigendecomposition, SVD, and matrix inversion \cite{demmel2007fast}.

\section{Background for general symmetry groups}
\label{sec:general_groups}

In this Section, we discuss the general problem of finding the ground state energy, ground state, and performing time evolution under a Hamiltonian $H$ on a finite-dimensional Hilbert space which is invariant under some representation $R$ of a symmetry group $G$. Some of these problems have previously been considered in a variety of special cases\changevfour{~\cite{PhysRevA.65.032325,somma2006efficient,zeier2011symmetry,shammah2018open,lowfermions2022}}. Here, we provide background on the general techniques used in solving these problems using a unified framework. \changevfour{We also give an equivalent formulation in Appendix~\ref{sec:general_tn_appendix} in the language of tensor networks.}

We consider the *-subalgebra of operators invariant under $R$ to which $H$ belongs. We think of this subalgebra as a standalone *-algebra $X$, such that the embedding of $X$ into the full operator algebra defines a representation $A$ of $X$. The practical relevance of these considerations is when the size of the total Hilbert space grows exponentially with some scaling parameter $n$. The paradigmatic example is the Hilbert space of $n$ qubits. If there are enough symmetries, it can happen that the dimension $N(n)$ of $X$ only grows polynomially with $n$, in which case many properties can be calculated efficiently \cite{somma2006efficient}. This restriction of $X$ to a lower-dimensional subspace may more generally happen beyond systems symmetric in the sense of Definition~\ref{def:inv_op}. Due to this, for now we focus explicitly on $X$ and $A$, rather than on $G$ and its representation $R$; we will discuss the connection of our results to $G$ and $R$ more specifically at the end of this Section. 

For the various algorithms we now consider, we will assume that different properties of $X$ and $A$ are known. In the course of proving Lemma~\ref{lem:general_gs}, we give an algorithm for finding the ground state of $H$. Every finite-dimensional *-algebra is isomorphic to a direct sum of irreducible blocks, and every representation is isomorphic to a direct sum of irreducible representations. That is, there is a block-diagonal orthonormal basis $\ket{\lambda, q_\lambda, p_\lambda}$ of the vector space acted upon by $A$, where $\lambda$ labels an irrep of $X$, $q_\lambda$ labels a basis vector internal to $\lambda$, and $p_\lambda$ labels a basis vector in the multiplicity vector space of $\lambda$; this is the basis in which we compute the ground state of $H$ (for some arbitrary and fixed multiplicity label $p_{\lambda 0}$). To prove our theorem, we assume knowledge of the matrix elements:
\begin{equation}
F^{i,\lambda}_{q_\lambda,q_\lambda'} \coloneqq \bra{\lambda, q_\lambda, p_{\lambda 0}} A_i \ket{\lambda, q_\lambda', p_{\lambda 0}}.
\label{eq:mat_el_def}
\end{equation}

In the following lemmas, we assume we are given the Hamiltonian $H\in A(X)$ as $h\in X$ expressed in the preferred basis, such that
\begin{equation}
\label{eq:hamiltonian_basis}
H=\sum_i h_i A_i\;.
\end{equation}
\changevfour{Though a given $H$ may not be classically described in this exact basis, transformations into this basis are typically efficient when the number of irreps and irrep dimensions are polynomially sized. For example, a Pauli decomposition can be used to transform operators into the Pauli basis studied later in the permutation invariant setting.} We first focus on the case when we are interested in finding the ground state of some representation of such a Hamiltonian, in a basis where the action of the representation is known.
\begin{lemma}[Finding the ground state of symmetric Hamiltonians]
Consider a subalgebra $A(X)$, and assume that the matrix elements $F^{i,\lambda}_{q_\lambda,q_\lambda'}$ are known as discussed above. Then the ground state energy and ground state of $H$ in the $\ket{\lambda,q_\lambda,p_{\lambda 0}}$ basis can be found in time
$\operatorname{O}\left(n_\lambda^2 n_q^4\right)$, where $n_\lambda$ are the number of irreps of $X$ and $n_q$ the maximum irrep dimension.
\label{lem:general_gs}
\end{lemma}
\begin{proof}
For each $\lambda$, we compute the operator with indices:
\begin{equation}
\label{eq:hamiltonian_reduced_basis}
\hat h_{q_\lambda,q_\lambda'}^\lambda \coloneqq \sum_i h_i F_{q_\lambda,q_\lambda'}^{i,\lambda}.
\end{equation}
This takes time $\operatorname O\left(\dim\left(X\right)^2\right) = \operatorname{O}\left(n_\lambda^2 n_q^4\right)$.
Note that in the $\ket{\lambda,q_\lambda,p_\lambda}$ basis, $H$ has a block diagonal form. Furthermore, as $p_\lambda$ labels isomorphic copies of irreps, we can find the ground state by fixing $p_{\lambda 0}$ WLOG. Namely, the ground state energy is given by:
\begin{equation}
    \gse(H) = \min_\lambda \gse(\hat h^\lambda)\;,
\end{equation}
where:
\begin{equation}
    \gse(\hat h^\lambda)\coloneqq\min_{\ket{\psi}}\bra{\psi}\hat h^\lambda\ket{\psi}.
\end{equation}
Furthermore, let
\begin{equation}
\lambda_{\text{min}} \coloneqq \operatorname{argmin}_\lambda \gse(\hat h^\lambda)
\end{equation}
and
\begin{equation}
\ket{\psi^\ast} \coloneqq \operatorname{argmin}_{\ket{\psi}}\bra{\psi}\hat h^{\lambda_\text{min}}\ket{\psi}\;.
\end{equation}
Then, for any $p$,
\begin{equation}
\ket{\lambda_\text{min},\psi^\ast,p}
\end{equation}
is a ground state in the $\ket{\lambda,q_\lambda,p_\lambda}$ basis. The dimension of $\hat h^\lambda$ is $\dim_X(\lambda)\times\dim_X(\lambda)$, and thus and thus diagonalizing it takes time $\operatorname{O}\left(n_q^\omega\right)$. As we need to do this for $\operatorname{O}(n_\lambda)$ many values of $\lambda$, the total runtime for finding $\ket{\psi^*}$ is $\operatorname{O}\left(n_\lambda n_q^\omega\right)$.
So the runtime of the computation in Eq.~\eqref{eq:hamiltonian_reduced_basis} dominates, and the overall runtime is $\operatorname{O}\left(n_\lambda^2 n_q^4\right)$.
\end{proof}

We now show that the dynamics of an initial state $\rho$ under equivariant unitaries can be classically simulated even if $\rho\neq A\left(X\right)$. This generalizes similar \changevfour{approaches taken in \cite{PhysRevA.65.032325} and \cite{lowfermions2022} for both free fermions and fermionic systems with particle number symmetry}. The Lemma as written assumes the initial state is given in a basis where the multiplicity indices have been traced out; in Section~\ref{sec:permutation_invariance} we will give concrete examples on how one can compute this efficiently under a variety of input models.
\begin{lemma}[Simulating equivariant dynamics]
Let
\begin{equation}
    O=\sum\limits_i o_i A_i
\end{equation}
be a projective measurement and
\begin{equation}
    U=\sum\limits_i u_i A_i
\end{equation}
a unitary operator. Let $\rho$ be a quantum state, and let $\rho_{(\lambda, q_\lambda, p_\lambda), (\lambda', q_\lambda', p_\lambda')}$ be the coefficients of the state in the $\ket{\lambda,q_\lambda,p_\lambda}$ basis.
Assume we are given the reduced density matrix
\begin{equation}
\label{eq:rho_tilde}
    \tilde\rho^\lambda_{q_0, q_1}\coloneqq \sum_{p_\lambda} \rho_{(\lambda, q_\lambda, p_\lambda), (\lambda, q_\lambda', p_\lambda)}\;.
\end{equation}
Then,
\begin{equation}
    \ell\left(\rho\right)=\tr\left(OU\rho U^\dagger\right)
    \label{eq:single_loss}
\end{equation}
can be computed in time $\operatorname{O}\left(n_\lambda^2 n_q^4\right)$.
\label{lem:sym_dyn}
\end{lemma}
\begin{proof}
After going to the $\ket{\lambda,q_\lambda,p_\lambda}$ basis, $O$ and $U$ become block diagonal, so we can compute $\tr\left(OU\rho U^\dagger\right)$ by restricting to fixed $p_{\lambda0}$.
To this end, we compute
\begin{equation}\label{eq:tilde_o}
\tilde o_{q_\lambda,q_\lambda'}^\lambda \coloneqq \sum_i o_i F_{q_\lambda,q_\lambda'}^{i,\lambda}.
\end{equation}
and
\begin{equation}\label{eq:tilde_u}
\tilde u_{q_\lambda,q_\lambda'}^\lambda \coloneqq \sum_i u_i F_{q_\lambda,q_\lambda'}^{i,\lambda}.
\end{equation}
Given the matrix elements $F_{q_\lambda,q_\lambda'}^{i,\lambda}$, this takes time $\operatorname O\left(\dim\left(X\right)^2\right) = \operatorname{O}\left(n_\lambda^2 n_q^4\right)$.
Then we have
\begin{equation}
\ell(\rho) = \tr\left(OU\rho U^\dagger\right) = \sum_\lambda \tr\left(\tilde o^\lambda \tilde u^\lambda \tilde \rho^\lambda (\tilde u^\lambda)^\dagger\right)\;.
\end{equation}
The right-hand side can be computed in time $\operatorname{O}\left(n_\lambda n_q^\omega\right)$, as it is given by the matrix multiplication of $n_\lambda$ blocks each of size at most $n_q\times n_q$.
\end{proof}
In the above considerations, the group $G$ and representation $R$ do not directly enter. In practice, however, we might want to start with those two. The irreps of $X$ are in one-to-one correspondence with those of $R$. By a simple corollary of the von Neumann bicommutant theorem, multiplicities of irreps in $X$ are the dimensions of the irreps of $G$, and the dimensions of the irreps of $G$ are the multiplicities of the irreps of $A$. We thus have that:
\begin{equation}
    \dim(X) = \sum_{\lambda} \dim_X(\lambda)^2 = \sum_{\lambda} \operatorname{mult}_R(\lambda)^2\;.
\end{equation}
Thus, the problems discussed above become classically tractable if the number $n_\lambda$ of irreps of $G$ with non-zero multiplicity in $R$, as well as the maximum multiplicity $n_q$ of an irrep $\lambda$ in $R$, are both polynomially small.

\section{Permutation invariance on qubits}
\label{sec:permutation_invariance}
We will now apply the general algorithms of Section~\ref{sec:general_groups} to the case where $G$ is given by $\mathrm{S}_n$ and $R$ is the representation on $n$ qubits acting by permutations,
\begin{equation}
\label{eq:permutation_representation}
\begin{multlined}
    R(\pi) \ket{i_1} \otimes \ket{i_2} \otimes \cdots \otimes \ket{i_n}\\
    = \ket{ i_{\pi^{-1}1}} \otimes \ket{ i_{\pi^{-1}2}} \otimes \cdots \otimes \ket{ i_{\pi^{-1}1}}\;.
\end{multlined}
\end{equation}
A straightforward basis for the algebra of invariant operators can be obtained by applying the Reynold's operator in Eq.~\eqref{eq:reynolds_definition} to the Pauli basis. Normalizing such that all operators $A_i$ are sums of unit norm Pauli terms, we obtain
\begin{equation}
    \begin{aligned}
        &A_{\bm i}=\frac{1}{i_1!i_x!i_y!i_z!}\cdot\\
        &\sum\limits_{\pi\in \mathrm{S}_n}R\left(\pi\right)\left(\sigma_1^{\otimes i_1} \otimes \sigma_x^{\otimes i_x} \otimes \sigma_y^{\otimes i_y} \otimes \sigma_z^{\otimes i_z}\right)R^{-1}\left(\pi\right),
    \end{aligned}
\end{equation}
for every $4$-tuple of positive integers
\begin{equation}
\bm{i}=(i_1, i_x, i_y, i_z): i_1+i_x+i_y+i_z=n\;,
\end{equation}
where $\sigma$ denote Pauli operators, and $\sigma_1=\operatorname{id}_2$.

Enumerating such 4-tuples shows that the dimension of the algebra $X$ is of order $\operatorname{O}(n^3)$. Therefore, the algorithms of Section~\ref{sec:general_groups} would reduce the task of finding the ground state (energy) of a Hamiltonian to a polynomial runtime in $n$, compared to exact diagonalization of the full Hamiltonian with a runtime exponential in $n$. Let us start with computing the ground state energy and ground state of a permutation-invariant Hamiltonian $H$ given as a vector $h_i$ in the basis of symmetrized Pauli monomials.
\begin{theorem}[Efficient classical computation of the ground state of a permutation-invariant Hamiltonian]
\label{thm:sym_gs}
The ground state and ground state energy of a permutation-invariant Hamiltonian on $n$ qubits, given as $h_i$ in the basis of symmetrized Pauli monomials above, can be computed in time $\operatorname O(n^7)$ via Lemma~\ref{lem:general_gs}. When applied to a single Hamiltonian for each system size $n$, the bottleneck of the algorithm is the computation of the matrix elements $F^{i,\lambda}_{q_\lambda, q_\lambda'}$ that takes $\operatorname{O}(n^7)$, but is independent on $h_i$. Once the matrix elements are computed, application of Lemma~\ref{lem:general_gs} only takes time $\operatorname{O}\left(n^6\right)$. The output ground state is given in the $\ket{\lambda, q_\lambda, p_{\lambda0}}$ Schur basis for a fixed $p_{\lambda0}$.
\end{theorem}
\begin{proof}
The block-diagonal basis of the permutation representation $R$ of $G=S_n$ is known as \emph{Schur basis}, which we will review in Appendix~\ref{app:schur_basis}.
From this it is easy to see that the number $n_\lambda$ of irreps with nonzero multiplicity is $n_\lambda = \operatorname{O}\left(n\right)$, and the maximal dimension $n_q$ of these irreps is $n_q=\operatorname{O}\left(n\right)$. Direct application of the algorithm in Lemma~\ref{lem:general_gs}, then leads to a runtime of $\operatorname{O}\left(n^6\right)$ for computing the ground state of $\mathrm{S}_n$-equivariant Hamiltonians in the Schur basis.
All that remains is the computation of the $\operatorname{O}(n^6)$ many matrix elements $F^{i,\lambda}_{q_\lambda, q_\lambda'}$ for a given $n$. To this end, we make use of the fact that the Schur basis $\ket{\lambda, q_\lambda, p_{\lambda0}}$ as well as the representation $A$ both have \emph{matrix product operator} (MPO) representations. We will here only give a sketch of the tensor-network computation, and refer to Appendix~\ref{sec:tensor_network_f_entries} for the details. The representations are of the form
\begin{equation}
\label{eq:schur_mpo_representations}
\langle\lambda, q_\lambda, \vec p_\lambda|\vec l\rangle=
\begin{tikzpicture}
\atoms{square}{{a0/p={0,0}}, {a1/p={1.2,0}}, {a2/p={3,0}}}
\draw (a0)--(a1) (a1)edge[mark={three dots,a}]++(0:0.4) (a2)edge[mark={three dots,a}]++(180:0.4);
\draw (a0-t)edge[ind=$l_0$]++(90:0.4) (a1-t)edge[ind=$l_1$]++(90:0.4) (a2-t)edge[ind=$l_{n-1}$]++(90:0.4);
\draw (a0-b)edge[ind=$(p_\lambda)_0$]++(-90:0.4) (a1-b)edge[ind=$(p_\lambda)_1$]++(-90:0.4) (a2-b)edge[ind=$(p_\lambda)_{n-1}$]++(-90:0.4);
\draw (a2-r)edge[ind=$q_\lambda$]++(0:0.4);
\end{tikzpicture}\;,
\end{equation}
of bond dimension $\operatorname O(n)$, for each fixed value of $\lambda$, and 
\begin{equation}
\label{eq:rep_mpo_representation}
\langle\vec l|A_{\mathbf i}|\vec l'\rangle=
\begin{tikzpicture}
\atoms{circ}{{a0/p={0,0}}, {a1/p={1.2,0}}, {a2/p={3,0}}}
\draw (a0)--(a1) (a1)edge[mark={three dots,a}]++(0:0.4) (a2)edge[mark={three dots,a}]++(180:0.4);
\draw (a0-t)edge[ind=$l_0'$]++(90:0.4) (a1-t)edge[ind=$l_1'$]++(90:0.4) (a2-t)edge[ind=$l_{n-1}'$]++(90:0.4);
\draw (a0-b)edge[ind=$l_0$]++(-90:0.4) (a1-b)edge[ind=$l_1$]++(-90:0.4) (a2-b)edge[ind=$l_{n-1}$]++(-90:0.4);
\draw (a2-r)edge[ind=$\mathbf i$]++(0:0.4);
\end{tikzpicture}
\end{equation}
of bond dimension $\operatorname O(n^3)$.
With this, the overlap in Eq.~\eqref{eq:mat_el_def} in the computation of the matrix elements of $F$ can be expressed as a tensor-network diagram:
\begin{equation}
\label{eq:fmatrix_computation}
\begin{gathered}
F^{i,\lambda}_{q_\lambda, q_\lambda'}
= \bra{\lambda, q_\lambda, p_{\lambda 0}} A_i \ket{\lambda, q_\lambda', p_{\lambda 0}}
\\
=
\begin{tikzpicture}
\atoms{circ}{{a0/p={0,0}}, {a1/p={1.2,0}}, {a2/p={3,0}}}
\atoms{square}{{b0/p={0,0.6}}, {b1/p={1.2,0.6}}, {b2/p={3,0.6}}}
\atoms{square}{{c0/p={0,-0.6}}, {c1/p={1.2,-0.6}}, {c2/p={3,-0.6}}}
\atoms{circ, tiny}{{x0/p={0,1.1}, lab={p=90:0.3, t=$(p_{\lambda0})_0$}}, {x1/p={1.2,1.1}, lab={p=90:0.3, t=$(p_{\lambda0})_1$}}, {x2/p={3,1.1}, lab={p=90:0.3, t=$(p_{\lambda0})_{n-1}$}}}
\atoms{circ, tiny}{{y0/p={0,-1.1}, lab={p=-90:0.25, t=$(p_{\lambda0})_0$}}, {y1/p={1.2,-1.1}, lab={p=-90:0.25, t=$(p_{\lambda0})_1$}}, {y2/p={3,-1.1}, lab={p=-90:0.25, t=$(p_{\lambda0})_{n-1}$}}}
\draw (a0)--(a1) (a1)edge[mark={three dots,a}]++(0:0.4) (a2)edge[mark={three dots,a}]++(180:0.4);
\draw (b0)--(b1) (b1)edge[mark={three dots,a}]++(0:0.4) (b2)edge[mark={three dots,a}]++(180:0.4);
\draw (c0)--(c1) (c1)edge[mark={three dots,a}]++(0:0.4) (c2)edge[mark={three dots,a}]++(180:0.4);
\draw (a0)--(b0) (a0)--(c0) (b0)--(x0) (c0)--(y0);
\draw (a1)--(b1) (a1)--(c1) (b1)--(x1) (c1)--(y1);
\draw (a2)--(b2) (a2)--(c2) (b2)--(x2) (c2)--(y2);
\draw (b2-r)edge[ind=$q_\lambda$]++(0:0.4) (c2-r)edge[ind=$q_\lambda'$]++(0:0.4) (a2-r)edge[ind=$i$]++(0:0.4);
\end{tikzpicture}\;.
\end{gathered}
\end{equation}
Contracting this tensor network from the left to the right (for a fixed $\lambda$ and $\vec p_{\lambda0}$) is the same as a sequence of $n-1$ vector-matrix multiplications of dimension $\operatorname O(n^5)$. This naive contraction already yields a polynomial runtime, namely $\operatorname O(n^{10})$ for each of $\operatorname O(n)$ contractions for each of $\operatorname O(n)$ values of $\lambda$, yielding $\operatorname O(n^{12})$ in total. However, in Appendix~\ref{sec:tensor_network_f_entries}, we make use of a band-diagonal structure of the matrices to reduce the cost of a single contraction from $\operatorname{O}(n^{10})$ to $\operatorname{O}(n^5)$, yielding $\operatorname{O}(n^7)$ in total.
In Appendix~\ref{sec:combinatorial_f}, we also give an alternative combinatorial way of calculating the matrix entries of $F$ which has a slightly worse runtime of $\operatorname{O}(n^{10})$ in total but is better parallelizable.
\end{proof}

\begin{remark}
The output of the classical algorithm in Theorem~\ref{thm:sym_gs} is the ground state in the Schur basis, which might not always be a useful description. However, from this output, we can be efficiently construct the ground state on $n$ qubits with a quantum computer via the Schur transform~\cite{bacon2006efficient, harrowschur,kirby2017practical}. Furthermore, we can construct an efficient classical representation of a ground state as a low-bond-dimension MPS as we discuss in Section~\ref{sec:mps_input_output} and Appendix~\ref{sec:time_evo_appendix}.
\end{remark}

We now consider an application of Lemma~\ref{lem:sym_dyn} to the symmetric group case.
\begin{theorem}[Efficient classical simulation of permutation-equivariant dynamics]
Consider the classical evaluation of $\ell\left(\rho\right)$ as in Eq.~\eqref{eq:single_loss} using the same assumptions as Lemma~\ref{lem:sym_dyn} with symmetry group $\mathrm{S}_n$. $\ell\left(\rho\right)$ can be calculated in time $\operatorname{O}\left(n^7\right)$.
\label{thm:perm_inv_dynamics}
\end{theorem}
\begin{proof}
As for the proof of Theorem~\ref{thm:sym_gs}, we can compute the matrix elements $F^{i,\lambda}_{q_\lambda, q_\lambda'}$ in time $\operatorname{O}(n^7)$.
The remainder is straight-forward application of Lemma~\ref{lem:sym_dyn} with $n_\lambda,n_q=\operatorname{O}(n)$.
\end{proof}

Theorems~\ref{thm:sym_gs} and~\ref{thm:perm_inv_dynamics} provide efficient end-to-end classical algorithms if the output ground state, or the input state (for computing the dynamics) are given in the Schur basis. Prior work has also shown similar results for Dicke states \cite{gegg2016efficient,shammah2018open}.
It may be desirable and more natural to process inputs and outputs of the algorithms that do not explicitly depend on the Schur basis.
In the following two subsections, we will discuss the algorithms for two different input and output formats, namely first as a quantum state, and second as a matrix product state. These results are summarized in Table \ref{tab:summary_table}

\begin{table*}[t]
\centering
\footnotesize
\begin{tabular}{lccl} \hline
\textbf{Task} & \textbf{Runtime} & \textbf{Ref.} & \textbf{Notes} \\ \hline
\multicolumn{4}{l}{\textit{Finding ground state:}} \\
$\; \; \; $ Output specified in Schur basis   &  $\operatorname{O}(n^7)$ & Thm. \ref{thm:sym_gs} & $\operatorname{O}(n^6)$ if $F$ matrix is precomputed.  \\ 
$\; \; \; $ Output matrix product state$^*$          & $\operatorname{O}(n^7)$ & Rem. \ref{rem:n6_tensor} & $\operatorname{O}(n^6)$ if $F$ matrix is precomputed.  \\ 
\multicolumn{4}{l}{\textit{Hamiltonian simulation:}} \\
$\; \; \; $ Any input quantum state$^{**}$      & $\operatorname{\tilde{O}}\left(\left\lvert O\right\rvert_\infty^2\epsilon^{-2}n^4\log\left(\delta^{-1}\right)\right)$  & Thm. \ref{thm:perm_eq_dyns_quant_inp} &  Added $\operatorname{O}(n^6)$ time to process observable. \\
$\; \; \; $  Input specified in Schur basis$^{***}$    &  $\operatorname{O}(n^6)$ & Thm. \ref{thm:perm_inv_dynamics} &             \\
$\; \; \; $ Matrix product state input      & $\operatorname{O}(\chi^\omega n^3+n^7)$ & Thm. \ref{thm:time_evo_mps}  & $\chi$ is the bond dimension.           \\
\hline
\multicolumn{4}{l}{$^*$Bond dimension of output $O(n)$.}\\
\multicolumn{4}{l}{$^{**}$Requires low-depth circuit and single qubit measurements to measure the quantum state.}\\
\multicolumn{4}{l}{$^{***}$Previous work provides efficient algorithms for symmetric Dicke states \cite{shammah2018open,gegg2016efficient}.}
\end{tabular}
\caption{Summary of runtimes and results for finding the ground state and simulating permutation invariant Hamiltonians on inputs of different forms. Throughout, $n$ refers to the number of qubits, $\delta$ is the probability of success and $\epsilon$ is the error of the outcome.}
\label{tab:summary_table}
\end{table*}

\subsection{Input and output as Quantum State}
In this subsection, we will describe how Theorems~\ref{thm:sym_gs} and~\ref{thm:perm_inv_dynamics} can be applied if the input or output is directly given as a quantum state on $n$ qubits.
Of course, this requires applying some quantum gates, measurements, and state preparations.
However, the necessary quantum computation is of shallow depth and its function is merely to prepare the state or to measure a classical shadow of the state, whereas the core of the computation is classical.

Specifically, we use the shallow-depth quantum circuit from Ref.~\cite{harrowschur,bacon2006efficient} implementing the Schur transformation $V_{STO}$ defined by:
\begin{equation}
    V_{STO}\ket{\lambda, q_\lambda, p_\lambda}=\ket{\bm{\lambda}}\ket{\bm{q_\lambda}}\ket{\bm{p_\lambda}}.
\end{equation}
Here, $\ket{\bm{\lambda}},\ket{\bm{q_\lambda}},\ket{\bm{p_\lambda}}$ are bitstring encodings of $\lambda,q_\lambda,p_\lambda$, respectively, in the computational basis, labeling the $\bm{\lambda}$ register, $\bm{q}$ register, and $\bm{p}$ register, respectively.

\begin{theorem}[Efficient classical preparation of the ground state of a permutation-invariant Hamiltonian]
Given a Hamiltonian in the symmetrized Pauli basis as in Theorem~\ref{thm:sym_gs}, we can prepare one of its ground states as a quantum state to an accuracy $\epsilon$ in classical runtime $\operatorname{O}\left(n^7\right)$ using a quantum circuit of depth $\operatorname{\tilde{O}}\left(n\operatorname{poly}\log\left(\epsilon^{-1}\right)\right)$.
\end{theorem}
\begin{proof}
We first use Theorem~\ref{thm:sym_gs} to get a representation of the ground state in the $\ket{\lambda, q_\lambda, p_\lambda}$-basis, given by $\lambda_{\text{min}}$ and $\ket{\psi^\ast}$.
We then prepare the quantum state
\begin{equation}
\ket{\bm{\lambda_{\text{min}}}} \left(\sum_{q_{\lambda_{\text{min}}}} \psi_{q_{\lambda_{\text{min}}}}^\ast \ket{\bm{q_{\lambda_{\text{min}}}}}\right)
\ket{\bm{p_{\lambda_{\text{min}} 0}}}
\end{equation}
for some $\bm{p_{\lambda_{\text{min}}0}}$ of choice.
This preparation takes time $\operatorname{O}\left(n\right)$.
The quantum ground state is then obtained by applying $V_{\text{STO}}$ to the above state, which can be implemented in time $\operatorname{\tilde{O}}\left(n\operatorname{poly}\log\left(\epsilon^{-1}\right)\right)$ on a quantum computer~\cite{harrowschur}.
\end{proof}

\begin{theorem}[Efficient classical simulation of permutation-equivariant dynamics on a given quantum state]
Consider the classical evaluation of $\ell\left(\rho\right)$ as in Eq.~\eqref{eq:single_loss} using the same assumptions as Lemma~\ref{lem:sym_dyn} with symmetry group $\mathrm{S}_n$, except instead of $\tilde{\rho}$ one is given $\rho$ as a quantum state. Then, $\ell\left(\rho\right)$ can be estimated to additive error $\epsilon$ with probability $1-\delta$ via $\operatorname{\tilde{O}}\left(\left\lVert O\right\rVert_\infty^2\epsilon^{-2}n^4\log\left(\delta^{-1}\right)\right)$ calls to a quantum computer each of depth $\operatorname{\tilde{O}}\left(n\operatorname{poly}\log\left(\epsilon^{-1}\right)\right)$, up to an additional time $\operatorname{O}\left(n^7\right)$ in classical processing.\label{thm:perm_eq_dyns_quant_inp}
\end{theorem}
\begin{proof}
    This follows from using the Schur transform~\cite{bacon2006efficient, harrowschur} to prepare $\tilde{\rho}$ \changevfour{(as defined in Eq.~\eqref{eq:rho_tilde})} from $\rho$, which can be done in time $\operatorname{\tilde{O}}\left(n\operatorname{poly}\log\left(\epsilon^{-1}\right)\right)$ on a quantum computer. \changevfour{Note that $\tilde{\rho}$ is a $\left\lceil\log_2\left(n^2\right)\right\rceil$-qubit state as it only acts on the registers labeled by irreps $\lambda$ and their basis vectors $q_\lambda$, i.e., the multiplicity register $p_\lambda$ has been traced out.} Given a source of $\rho$, then, a quantum computer can efficiently use classical shadow techniques~\cite{huang2020predicting} to \changevfour{output a classical description of $\tilde{\rho}$ capable of estimating $\ell\left(\rho\right)$ to additive error $\epsilon$ with probability $1-\delta$ using
    \begin{equation}
        N=\operatorname{\tilde{O}}\left(\left\lVert\tilde{u}^\dagger\tilde{o}\tilde{u}\right\rVert_{\text{shadow}}^2\epsilon^{-2}\log\left(\delta^{-1}\right)\right)
    \end{equation}
    samples of the state. Here, $\left\lVert\cdot\right\rVert_{\text{shadow}}$ is the shadow norm, and $\tilde{o}$ and $\tilde{u}$ are defined as in Eqs.~\eqref{eq:tilde_o} and~\eqref{eq:tilde_u}, respectively. As $\tilde{u}^\dagger\tilde{o}\tilde{u}$ is on $\left\lceil\log_2\left(n^2\right)\right\rceil$ qubits we have the general bound~\cite{huang2020predicting}
    \begin{equation}
        \left\lVert\tilde{u}^\dagger\tilde{o}\tilde{u}\right\rVert_{\text{shadow}}^2\leq n^4\left\lVert\tilde{u}^\dagger\tilde{o}\tilde{u}\right\rVert_\infty^2=n^4\left\lVert O\right\rVert_\infty,
    \end{equation}
    where $\left\lVert\cdot\right\rVert_\infty$ is the operator norm.} We have then reduced the problem to a purely classical one, for which the classical runtime is given by Theorem~\ref{thm:perm_inv_dynamics}.
\end{proof}
The sample complexity of this procedure could potentially be improved to $\operatorname{\tilde{O}}(\left\lVert O\right\rVert_\infty^2\epsilon^{-2}n^3\log\left(\delta^{-1}\right))$ as $\tilde{\rho}$ only has $n^3$ degrees of freedom. However, the block diagonal structure over irreps is lost when transforming to the bitstring encoding $\ket{\bm{\lambda}}\ket{\bm{q}}\ket{\bm{p}}$ via the Schur transform, and thus we arrive at the sample complexity given.

\subsection{Input and Output as Matrix Product States}
\label{sec:mps_input_output}
The previous subsection implemented algorithms with quantum states as outputs and inputs.
In that particular setting, the readout of a quantum state requires some form of quantum computation of low depth.
Here we describe alternative algorithms where we give the inputs and outputs purely classically as low-bond dimension MPS, which are discussed in more detail in Appendix~\ref{sec:time_evo_appendix}.
Let us start by giving the output of Theorem~\ref{thm:sym_gs} as an MPS.

\begin{theorem}
\label{thm:ground_state_mps}
The output of Theorem~\ref{thm:sym_gs} can be given purely classically as an MPS of bond dimension $\operatorname{O}(n)$.
\end{theorem}
\begin{proof}
The output of Theorem~\ref{thm:sym_gs} is a value of $\lambda$ and a vector $\rho_{q_\lambda}$.
An MPO representation of the ground state can be obtained by contracting $\rho$ with the MPO for the Schur basis in  Eq.~\eqref{eq:schur_mpo_representations},
\begin{equation}
\begin{tikzpicture}
\atoms{square}{{a0/p={0,0}}, {a1/p={1.2,0}}, {a2/p={3,0}}}
\atoms{circ, all, lab={t=$\rho$, p=0:0.3}}{r/p={3.8,0}}
\draw (a0)--(a1) (a1)edge[mark={three dots,a}]++(0:0.4) (a2)edge[mark={three dots,a}]++(180:0.4);
\draw (a0-t)edge[ind=$l_0$]++(90:0.4) (a1-t)edge[ind=$l_1$]++(90:0.4) (a2-t)edge[ind=$l_{n-1}$]++(90:0.4);
\draw (a0-b)edge[ind=$(p_\lambda)_0$]++(-90:0.4) (a1-b)edge[ind=$(p_\lambda)_1$]++(-90:0.4) (a2-b)edge[ind=$(p_\lambda)_{n-1}$]++(-90:0.4);
\draw (a2-r)--(r);
\end{tikzpicture}.
\end{equation}
As we mentioned around Eq.~\eqref{eq:schur_mpo_representations}, the bond dimension of this MPO is $\operatorname{O}(n)$.
We consider this MPS representation in more detail in Appendix~\ref{sec:time_evo_appendix}.
\end{proof}

Next, let us consider Theorem~\ref{thm:perm_inv_dynamics} with an MPS as classical input.
\begin{theorem}
\label{thm:time_evo_mps}
Consider a state $\psi$ represented as an MPS of bond dimension $\chi$.
The time evolved expectation value $\ell\left(\rho\right)$ with $\rho=\ket\psi\bra\psi$ as in Theorem~\ref{thm:perm_inv_dynamics} can be computed classically in runtime $\operatorname{O}(\chi^\omega n^3+n^7)$.

\end{theorem}
\begin{proof}
Denoting the MPS as tensor-network diagram,
\begin{equation}
\begin{tikzpicture}
\atoms{square,small}{{a0/p={-0.2,1.4}}, {a1/p={1.3,1.4}}, {a2/p={3.8,1.4}}}
\draw[mps] (a0)--(a1) (a1)edge[mark={three dots,a}]++(0:0.9) (a2)edge[mark={three dots,a}]++(180:0.5);
\draw (a0-t)edge[ind=$l_0$]++(90:0.4) (a1-t)edge[ind=$l_1$]++(90:0.4) (a2-t)edge[ind=$l_{n-1}$]++(90:0.4);
\end{tikzpicture}\;,
\end{equation}
the reduced density matrix $\tilde{\rho}^\lambda$ as defined in Eq.~\eqref{eq:rho_tilde}
is given by the overlap with twice Eq.~\eqref{eq:schur_mpo_representations}:
\begin{equation}
\begin{tikzpicture}
\atoms{square}{{a0/p={0,-0.3}}, {a1/p={1.2,-0.3}}, {a2/p={3,-0.3}}}
\atoms{square}{{b0/p={0,0.3}}, {b1/p={1.2,0.3}}, {b2/p={3,0.3}}}
\atoms{square,small}{{c0/p={0,0.8}}, {c1/p={1.2,0.8}}, {c2/p={3,0.8}}}
\atoms{square,small}{{d0/p={0,-0.8}}, {d1/p={1.2,-0.8}}, {d2/p={3,-0.8}}}
\draw (a0)--(a1) (a1)edge[mark={three dots,a}]++(0:0.4) (a2)edge[mark={three dots,a}]++(180:0.4);
\draw (b0)--(b1) (b1)edge[mark={three dots,a}]++(0:0.4) (b2)edge[mark={three dots,a}]++(180:0.4);
\draw (c0)--(c1) (c1)edge[mark={three dots,a}]++(0:0.4) (c2)edge[mark={three dots,a}]++(180:0.4);
\draw (d0)--(d1) (d1)edge[mark={three dots,a}]++(0:0.4) (d2)edge[mark={three dots,a}]++(180:0.4);
\draw (a0)--(b0) (a0)--(d0) (b0)--(c0);
\draw (a1)--(b1) (a1)--(d1) (b1)--(c1);
\draw (a2)--(b2) (a2)--(d2) (b2)--(c2);
\draw (a2-r)edge[ind=$q_\lambda$]++(0:0.4) (b2-r)edge[ind=$q_\lambda'$]++(0:0.4);
\end{tikzpicture}\;.
\end{equation}
The naive contraction of this tensor network from left to right would need runtime $\operatorname{O}(n^8\chi^4)$.
However, using that the MPO in the middle is band diagonal, the runtime reduces to $\operatorname{O}(n^4\chi^4)$.
We discuss this contraction in more detail in Appendix~\ref{sec:time_evo_appendix}.
After we have obtained $\rho$, we can proceed as in Theorem~\ref{thm:perm_inv_dynamics}.
\end{proof}

\begin{remark} \label{rem:n6_tensor}
    If the matrix $F$ in Theorem~\ref{thm:sym_gs} is pre-computed, then this runtime is reduced to $\operatorname{O}(\chi^\omega n^3+n^6)$.
\end{remark}

\subsection{Applications to machine learning}

Though it is widely believed that general quantum machine learning models are more expressive than their classical counterparts~\cite{liu2021rigorous}, it is known that generic quantum machine learning models are untrainable~\cite{mcclean2018barren,cerezo2021cost,marrero2021entanglement,anschuetz2022critical,napp2022quantifying,larocca2022diagnosingbarren,larocca2021,anschuetzkiani2022}. This gives a strong motivation for constructing quantum machine learning models from symmetry-equivariant time dynamics~\cite{equivariant} which are believed to be efficiently trainable~\cite{larocca2022diagnosingbarren,larocca2021}. This was formally proven for $\operatorname{S}_n$-equivariant models in Ref.~\cite{schatzki2022theoretical}.

As a specific application of our results, we now consider one of the learning problems for which a variational quantum algorithm was given in Ref.~\cite{schatzki2022theoretical}. \changevtwo{We emphasize that here, just as in Lemma~\ref{lem:sym_dyn}, we do not require that the input states $\rho_i$ respect the symmetries of the model.} For simplicity we assume the quantum input setting of Theorem~\ref{thm:perm_eq_dyns_quant_inp}, but a similar theorem can be given for MPS inputs.
\begin{corollary}[Efficient classical simulation of permutation-invariant models]
Consider a binary classification problem with labels $y_i\in\left\{-1,1\right\}$ and empirical loss
\begin{equation}
    \hat{\mathcal{L}}\left(\bm{\theta}\right)=-\frac{1}{M}\sum\limits_{i=1}^M y_i\ell_{\bm{\theta}}\left(\rho_i\right),
\end{equation}
where $\ell_{\bm{\theta}}\left(\rho_i\right)$ is as in Eq.~\eqref{eq:single_loss} with a $\bm{\theta}$-dependent $U$ and otherwise the same assumptions as in Theorem~\ref{thm:perm_eq_dyns_quant_inp}. $\hat{\mathcal{L}}$ can be estimated to additive error $\epsilon$ at $P$ points in time
\begin{equation}
    \operatorname{\tilde{O}}\left(M\left\lVert O\right\rVert_\infty^2\epsilon^{-2}n^5\log\left(\frac{P}{\delta}\right)+MPn^{\omega+1}+n^7\right)
\end{equation}
with total probability of success at least $1-\delta$.\label{cor:mac_learn}
\end{corollary}
\begin{proof}
This follows immediately from Theorem~\ref{thm:perm_eq_dyns_quant_inp} with $\delta\to\frac{\delta}{P}$ by the union bound.
\end{proof}
\changevtwo{Corollary~\ref{cor:mac_learn} implies that the loss of these models can be estimated completely classically when the states $\rho_i$ are given as certain classical shadows or matrix product state descriptions, even if they do not respect the symmetries of the model.} As a point of comparison, consider the runtime of using a variational quantum algorithm to perform this binary classification task. Assume the variational circuits are of depth $\operatorname{\Omega}\left(n^3\right)$ as required in Theorem 3 of Ref.~\cite{schatzki2022theoretical} to ensure convergence. Then---taking $\operatorname{\Omega}\left(\left\lVert O\right\rVert_\infty^2\epsilon^{-2}\right)$ samples for each measurement to achieve an overall shot noise of $\operatorname{O}\left(\epsilon\right)$---this yields an overall runtime of $\operatorname{\Omega}\left(MP\left\lVert O\right\rVert_\infty^2\epsilon^{-2}n^3\right)$. For $P$ sufficiently large, compare this to the time $\operatorname{O}\left(MPn^{\omega+1}\right)$ algorithm found for the classical algorithm where, even if quantum states are given as input, a classical shadow representation can be measured in quantum depth only $\operatorname{\tilde{O}}\left(n\operatorname{poly}\log\left(\epsilon^{-1}\right)\right)$. Unlike the quantum algorithm, this algorithm can be parallelized over irreps (i.e., over $n_\lambda$) easily, giving an effective runtime $\operatorname{O}\left(MPn^\omega\right)=\operatorname{o}\left(MPn^3\right)$. Even for $P$ small, given many QPUs capable of running depth $\sim n$ quantum circuits, the classical algorithm parallelizes more effectively than the quantum algorithm as the required shadow tomography can be straightforwardly parallelized over shots.

\section{Conclusion}

We have specified a general framework for classically simulating highly symmetric quantum systems. Specializing to the symmetric group, we showed that these techniques yield an efficient classical algorithm for finding the ground state of quantum systems obeying an $\mathrm{S}_n$ symmetry, evaluating dynamics, and simulating $\mathrm{S}_n$-equivariant quantum machine learning models. We hope that this framework sets the foundations for the future study of classical characterizations of symmetric quantum systems.  \changevfour{Potential applications include the physical analysis of symmetric quantum systems \cite{zeier2011symmetry,chase2008collective,kirton2018superradiant,shankar2017steady} and implementation of algorithmic primitives for measuring entanglement \cite{horodecki2009quantum}, performing measurements with distributed sensors \cite{zhang2021distributed}, and learning or testing quantum states \cite{alicki1988symmetry,o2021learning,childs2007weak}. }

A major remaining open question is to what extent the classical algorithms we have proposed for the case of qubit permutation invariance carry over to the more general highly symmetric many-body quantum systems. The approach we presented in Section~\ref{sec:general_groups} already works for general symmetry groups and representations.
However, the computation of the matrix elements of $F$, as well as the representations of ground and initial states as matrix product states or quantum states was only described for the case of qubit permutation invariance. The perhaps most obvious generalization to other symmetry groups of representations is to consider permutation-invariant systems consisting of $n$ $d$-dimensional qudits instead of qubits.
We sketch this generalization in Appendix~\ref{sec:qudits}.
As it turns out, all algorithms we have proposed generalize with polynomial runtimes, however, the exponents become very large for large $d$.
Specifically, the computation of the matrix elements $F$, which is the bottleneck for many of our algorithms, takes time $\operatorname{O}(n^{2d^2-1})$ in this case.

It should be noted that the rather large exponents of the polynomial runtimes ($\operatorname{O}(n^7)$ for qubits and $\operatorname{O}(n^{2d^2-1})$ for qudits) do not come from the method we use, but from the way we phrase the problem itself:
Specifying a symmetric Hamiltonian alone needs $\operatorname{O}(n^3)$ data for qubits, and $\operatorname{O}(n^{d^2-1})$ data for qudits, so when measuring the runtime in the size of the input data, the scaling is just a little more than quadratic.
For more efficient algorithms we would need to restrict the input to a subset of permutation-invariant Hamiltonians/operators, for example to Hamiltonians that are sparse in the symmetrized Pauli basis, or $k$-local Hamiltonians for constant $k$.

Our results indicate that care must be taken when constructing quantum algorithms to solve problems described by strongly symmetric quantum systems to ensure no equally efficient classical algorithm exists, as was demonstrated in Corollary~\ref{cor:mac_learn}. \changevfour{One approach to avoid classical simulation methods is to consider smaller symmetry groups as illustrated in Figure~\ref{fig:illustrative_stuff}(a). Indeed, complexity theory arguments~\cite{10.1145/3519935.3520067,10.1145/3519935.3520052} suggest that solving the ground state problem for systems with sufficiently small symmetry groups---such as translational invariance---is expected to be difficult.} \changevfour{It is also} reasonable to expect quantum algorithms for strongly symmetric quantum systems to offer polynomial speedups in appropriately selected scenarios, such as has been demonstrated in Ref.~\cite{PRXQuantum.4.020338}. Notably, the tensor contraction algorithms outlined in Theorem~\ref{thm:sym_gs} currently have a runtime of $\operatorname{O}(n^7)$, which in terms of the number of qubits $n$ is not as optimal as quantum algorithms. We hope this work sparks further investigation into how problems involving symmetric quantum systems may produce a practical use case for quantum algorithms.

\begin{acknowledgments}
     Code implementing some of the algorithms discussed is available at \url{https://github.com/bkiani/symmetric_hamiltonians}. \changevtwo{We thank Quynh T.\ Nguyen for pointing out connections to tensor networks. We thank Marco Cerezo, Martin Larocca, Frederic Sauvage, and Louis Schatzki for their helpful feedback on the first draft of this manuscript.} E.R.A.\ is supported by STAQ under award NSF Phy-1818914. A.B., B.T.K., and S.L.\ were supported by DARPA under the RQMLS program. A.B.\ is supported by the DFG (CRC 183), BMWK (PlanQK, EniQmA), and the Munich Quantum Valley (K-8).
\end{acknowledgments}

\bibliographystyle{quantum}
\bibliography{main}

\onecolumn
\appendix
\pagebreak

\section{The Schur basis}
\label{app:schur_basis}

In the presence of permutation invariance, the action of operations can be fully understood by analyzing a much smaller subspace of the larger Hilbert space. To precisely understand the form of that subspace, we turn to the Schur--Weyl decomposition of $n$ qubits into subspaces corresponding to irreducible representations of the symmetric and unitary groups labeled by Young diagrams. Schur--Weyl duality offers a means to perform this decomposition by considering the natural representations of the permutation group and $n$-fold unitary group acting on $n$ qubits  \cite{harrowschur,ping2002group}. To describe the Schur basis and the resulting Schur transform, first we note the natural action of a permutation operation $R(\pi)$ acting on qubits:
\begin{equation}
    R(\pi) \ket{i_1} \otimes \ket{i_2} \otimes \cdots \otimes \ket{i_n} = \ket{ i_{\pi^{-1}1}} \otimes \ket{ i_{\pi^{-1}2}} \otimes \cdots \otimes \ket{ i_{\pi^{-1}1}}
\end{equation}
as in the main text.

Similarly, a unitary $U \in \mathcal{U}(2)$ acting as the $n$-fold product $Q(U)$ takes the form
\begin{equation}
    Q(U) \ket{i_1}\otimes \ket{i_2} \otimes\cdots \otimes\ket{i_n} = U\ket{ i_1}\otimes U \ket{i_2} \otimes\cdots \otimes U\ket{ i_n}.
\end{equation}

Schur--Weyl duality takes advantage of the fact that $Q(\cdot)$ and $R(\cdot)$ are each others' commutants, stating that the subspace of $(\mathbb{C}^2)^{\otimes n}$ decomposes as 
\begin{equation}
    Q(U) R(\pi) \cong \bigoplus_\lambda \rho_\lambda(U) \otimes \sigma_\lambda(\pi),
\end{equation}
where $\lambda$ runs over the set of partitions of $n$ into at most two elements, and $\rho_\lambda(\cdot)$ and $\sigma_\lambda(\cdot)$ are irreducible representations of the unitary group $\mathcal{U}(2)$ and the symmetric group $\mathit{S}_n$, respectively. Note that irreps of both of these groups are indexed by partitions. More generally, for the space $(\mathbb{C}^d)^{\otimes n}$ of $n$ qudits of dimension $d$, the $\lambda$ would span over partitions of $n$ into at most $d$ elements. Partitions can equivalently be enumerated by Young diagrams. For example for the setting of $4$ qubits, we have the $3$ Young diagrams below that appear in the decomposition above:
\begin{align*}
\lambda = (4,0):& \begin{ytableau}
 \none[] & & & &
\end{ytableau}, \\ \\
\lambda = (3,1):&
\begin{ytableau}
 \none[]&&&\cr \none[] &
\end{ytableau},\\ \\
\lambda = (2,2):&\begin{ytableau}
 \none[] &&\cr \none[]&&
\end{ytableau}.
\end{align*}

A consequence of the above is that there exists a basis indexed by $\ket{\lambda,q_\lambda,p_\lambda}$ called the \emph{Schur basis} where the actions of $Q(\cdot)$ and $R(\cdot)$ are separated~\cite{harrowschur}:
\begin{align}
    Q(U)\ket{\lambda,q_\lambda,p_\lambda} &= \rho_\lambda(U)\ket{\lambda,q_\lambda,p_\lambda}, \\
    R(\pi)\ket{\lambda,q_\lambda,p_\lambda} &= \sigma_\lambda(\pi)\ket{\lambda,q_\lambda,p_\lambda},
\end{align}
where we have implicitly projected onto the subspace indexed by $\lambda$. Here, $\rho_\lambda(U)$ and $\sigma_\lambda(\pi)$ act only on the $q_\lambda$ and $p_\lambda$ space, respectively. $\rho_\lambda(U)$ and $\sigma_\lambda(\pi)$ are respectively the linear transformations corresponding to the irreducible representations of $\mathrm{U}_d$ and $\mathrm{S}_n$ for the irreducible representation indexed by $\lambda$. The above also presents a useful fact about permutation invariance. Namely, such an operation will act invariantly on the permutation register $\ket{p_\lambda}$ thus significantly reducing the degrees of freedom of a problem. The Schur transform $U_{Sch}$ is a unitary transformation that acts as a change of basis from the computational to the Schur basis described above. The Schur transform can be efficiently implemented on a quantum computer running in time $\operatorname{\tilde{O}}\left(n \text{poly}(d,\log n,1/\epsilon)\right)$ for error $\epsilon$ on qudit systems of local dimension $d$~\cite{harrowschur}. We follow the notation of Ref.~\cite{harrowschur}:
\begin{equation}
    \ket{\lambda,q_\lambda,p_\lambda} = \sum_{i_1,i_2,\dots,i_n = 0}^{d-1} [U_{Sch}]_{i_1,i_2,\dots,i_n}^{\lambda,q_\lambda,p_\lambda} \ket{i_1} \ket{i_2} \cdots \ket{i_n}.  
\end{equation}

As noted in the main text, the total degrees of freedom reduces to $\binom{n+3}{3}$ in settings with permutation invariance. To see this, note that the dimension of the $\ket{q_\lambda}$ register for a partition $(a,b)$ is equal to $a-b+1$. Therefore, we have
\begin{equation}
    \text{DOF} = \sum_{k=0}^{\lfloor n / 2 \rfloor} \left(2k+1 + n - 2\left\lfloor \frac{n}{2}\right\rfloor \right)^2 = \binom{n+3}{3}
\end{equation}
degrees of freedom. A similar calculation can be performed via a stars-and-bars counting argument. The above is also enumerated by the tetrahedral numbers~\cite{oeis}.

\begin{figure}[t]
    \centering
    \includegraphics{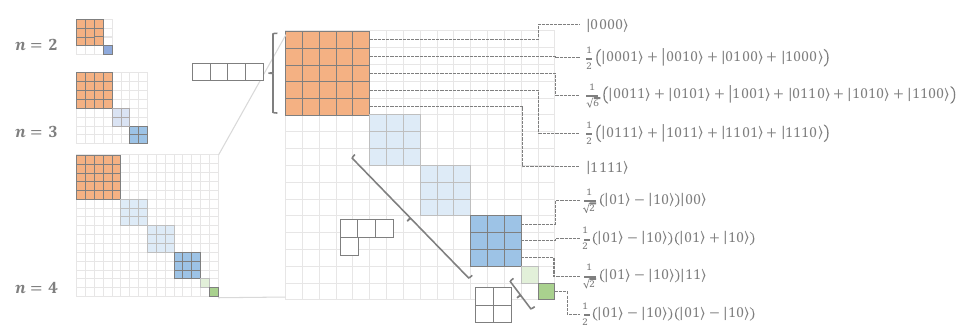}
    \caption{Graphical depiction of Schur decomposition for $n=4$ qubits. There are three Young diagrams of at most two rows for $4$ qubits. Due to the presence of permutation invariance, we can restrict attention to the darker colored subspaces which correspond to a single subspace over the multiplicity of the permutation irreps. To project onto this darker colored subspace, we use the Young symmetrizer (Eq.~\eqref{eq:young_symmetrizer}). }
    \label{fig:schur_decomposition_visual}
\end{figure}

To expand and manipulate individual basis states indexed by the $\ket{q_\lambda}$ register, one can use the Young symmetrizer $\Pi_{p_\lambda}$ to project onto an explicit basis for each $\lambda$ \cite{harrowschur,fulton1997young}. Here, $p_\lambda$ is a particular Young tableau for the Young diagram $\lambda$. The Young symmetrizer projects onto a subspace isomorphic to the subspace spanned by $\ket{q_\lambda}$:
\begin{equation}
\label{eq:young_symmetrizer}
    \Pi_{p_\lambda} = \frac{\operatorname{dim}(\lambda)}{n!} \left( \sum_{c \in \text{Col}(p_\lambda)} \operatorname{sgn}(c) \changevthree{R}(c) \right) \left( \sum_{r \in \text{Row}(p_\lambda)}  \changevthree{R}(r) \right),
\end{equation}
where $\text{Row}(p_\lambda)$ and $\text{Col}(p_\lambda)$ are the set of permutations which permute integers within only rows and columns of the Young tableau $p_\lambda$, respectively \cite{harrowschur,fulton1997young}. An example of the basis found via application of the Young symmetrizer is shown in Figure~\ref{fig:schur_decomposition_visual}. Throughout our study, we consider the Young tableau formed by filling entries in order first column-wise and then row-wise to be the ``canonical" basis that we study. As an example, for 4 qubits, there are the following Young tableaus in our ``canonical" basis:
\begin{align}
\label{eq:standard_young_tableaux}
\begin{ytableau}
 1 & 2 & 3 & 4 
\end{ytableau}, \;\;\;\;
\begin{ytableau}
 1 & 3 & 4 \cr 2 
\end{ytableau}, \; \; \; \; 
\begin{ytableau}
 1& 3 \cr 2 & 4
\end{ytableau}.
\end{align}

\section{General formalism using tensor-network diagrams}
\label{sec:general_tn_appendix}

In this and the two subsequent appendices, we discuss the tensor-network based methods used in Section~\ref{sec:permutation_invariance} of the main text in more detail. As a preparation, we will here recap the general formalism from Section~\ref{sec:general_groups} in a tensor-network language.

An associative algebra on a vector space $B$ is determined by a linear multiplication map $B\otimes B\rightarrow B$, whose structure coefficients form a 3-index tensor,
\begin{equation}
\label{eq:algebra_tensor}
\begin{tikzpicture}
\atoms{circ}{0/}
\draw (0)edge[mark={arr,f,s},ind=$a$]++(135:0.8) (0)edge[mark={arr,f,s},ind=$b$]++(45:0.8) (0)edge[ind=$c$]++(-90:0.8);
\end{tikzpicture}\;.
\end{equation}
Here, $a$ and $b$ label basis vectors at the input of the multiplication map, and $c$ is at the output. In our setting, there are three different relevant algebras, namely $X$, the group algebra of $S_n$, and the group algebra of $SU(2)$. A representation of the algebra on a vector space $V$ is determined by a linear map $V\otimes B\rightarrow V$, whose structure coefficients are again a 3-index tensor,
\begin{equation}
\begin{tikzpicture}
\atoms{square}{0/}
\draw (0)edge[ind=$a$]++(90:0.8) (0)edge[nqubit,mark={arr,f,s},ind=$l$]++(180:0.8) (0)edge[nqubit,ind=$m$]++(0:0.8);
\end{tikzpicture}\;.
\end{equation}
As depicted, we will often use different line styles for different vector spaces such as $B$ and $V$; here we made the $V$-indices thick. Again, there are three different relevant representations: The representation $A$ of $X$, the representation $R$ of $S_n$, and the representation $(\frac12)^{\otimes n}$ of $SU(2)$, all of which act on the $n$-qubit Hilbert space. The statement that a Hamiltonian (or an arbitrary operator) $H$ is invariant under the symmetry representation $R$ can be written as
\begin{equation}
\label{eq:h_invariant}
\begin{tikzpicture}
\atoms{square}{0/lab={t=$R$,p=-90:0.4}}
\atoms{square,lab={t=$H$,p=-90:0.4}}{h/p={-1,0}}
\draw (0)edge[ind=$a$]++(90:0.8) (h-l)edge[nqubit,ind=$l$,mark={arr,f,s}]++(180:0.4) (h-r)edge[mark={arr,f,e},nqubit](0) (0)edge[nqubit,ind=$m$]++(0:0.8);
\end{tikzpicture}
=
\begin{tikzpicture}
\atoms{square}{0/lab={t=$R$,p=-90:0.4}}
\atoms{square,lab={t=$H$,p=-90:0.4}}{h/p={1,0}}
\draw (0)edge[ind=$a$]++(90:0.8) (0)edge[nqubit,mark={arr,f,s},ind=$l$]++(180:0.8) (h-r)edge[nqubit,ind=$m$]++(0:0.4) (h-l)edge[nqubit,mark={arr,f,s}](0);
\end{tikzpicture}\;.
\end{equation}
The commutativity of the representation $A$ of $X$ with the representation $R$ of $S_n$ reads
\begin{equation}
\label{eq:commutant}
\begin{tikzpicture}
\atoms{square}{0/lab={t=$R$,p=-90:0.4}}
\atoms{square,lab={t=$A$,p=-90:0.4}}{h/p={-1,0}}
\draw (0)edge[ind=$a$]++(90:0.8) (h)edge[ind=$b$]++(90:0.8) (h-l)edge[nqubit,mark={arr,f,s},ind=$l$]++(180:0.4) (h-r)edge[mark={arr,f,e},nqubit](0) (0-r)edge[nqubit,ind=$m$]++(0:0.4);
\end{tikzpicture}
=
\begin{tikzpicture}
\atoms{square}{0/lab={t=$R$,p=-90:0.4}}
\atoms{square,lab={t=$A$,p=-90:0.4}}{h/p={1,0}}
\draw (0)edge[ind=$a$]++(90:0.8) (h)edge[ind=$b$]++(90:0.8) (0-l)edge[nqubit,mark={arr,f,s},ind=$l$]++(180:0.4) (h-r)edge[nqubit,ind=$m$]++(0:0.4) (h-l)edge[nqubit,mark={arr,f,s},](0);
\end{tikzpicture}\;.
\end{equation}
Since by construction $A$ spans the whole commutant of $R$, $H$ can be parametrized by some $X$-element $h$ as
\begin{equation}
\label{eq:H_from_H}
\begin{tikzpicture}
\atoms{square,lab={t=$H$,p=-90:0.4}}{h/}
\draw (h-l)edge[mark={arr,f,s},nqubit,ind=$l$]++(180:0.4) (h-r)edge[nqubit,ind=$m$]++(0:0.4);
\end{tikzpicture}
=
\begin{tikzpicture}
\atoms{square}{{h/lab={t=$h$,p=90:0.4},p={0,0.8}}}
\atoms{square,lab={t=$A$,p=-90:0.4}}{a/}
\draw (a)--(h) (a-l)edge[mark={arr,f,s},nqubit,ind=$l$]++(180:0.4) (a-r)edge[nqubit,ind=$m$]++(0:0.4);
\end{tikzpicture}\;.
\end{equation}
Eq.~\eqref{eq:h_invariant} then follows through Eq.~\eqref{eq:commutant}.

Next, let us describe how to obtain the ground state energy and ground state via Theorem~\ref{thm:sym_gs} from the main text using the matrix $F$. To this end, we use a central theorem of representation theory \cite{Davidson1996}: 1) Every finite-dimensional *-algebra (which includes group algebras and the algebra $X$) is isomorphic to a direct sum of full matrix algebras, and 2) every (unitary) representation is isomorphic to a direct sum of irreducible representations, which are projections onto one of the matrix-algebra summands. The two isomorphisms, which we call $O$ and $S$, respectively, are unitary matrices,
\begin{equation}
\begin{tikzpicture}
\atoms{triang,rot=90,lab={t=$O$,p=-120:0.35}}{0/}
\draw[irrep] (0-mb)edge[ind=$\lambda$]++(0:0.5);
\draw (0-ct)edge[ind=$a$]++(180:0.5);
\draw[qind] (0-cl)edge[ind=$q$]++(-10:0.5) (0-cr)edge[ind=$q'$]++(10:0.5);
\end{tikzpicture}
\;,\qquad
\begin{tikzpicture}
\atoms{triang,rot=90,lab={t=$S$,p=-120:0.35}}{0/}
\draw[irrep] (0-mb)edge[ind=$\lambda$]++(0:0.5);
\draw[nqubit] (0-ct)edge[ind=$l$]++(180:0.5);
\draw[qind] (0-cr)edge[ind=$q$]++(10:0.5);
\draw (0-cl)edge[pind,postaction={ind=$p$}]++(-10:0.5);
\end{tikzpicture}\;.
\end{equation}
Here, the index $a$ labels basis elements of the algebra, the \emph{irrep index} $\lambda$ labels irreducible representations of the algebra, \emph{internal indices} like $q$ and $q'$ label basis vectors within an irreducible representation $\lambda$, and the \emph{multiplicity index} $p$ labels the different copies of a fixed irreducible representation $\lambda$, if $\lambda$ occurs multiple times in the representation. Note that indices with different bond dimensions were drawn with different line styles, such as thick for irrep, dotted for internal, and ticked for multiplicity indices. Note that the bond dimension of the internal and multiplicity indices can depend on the value $\lambda$ of the nearby irrep index, so the diagrams are not classical tensor-network notation but a slightly generalized version thereof. In this notation, the isomorphism can be written as
\begin{equation}
\begin{tikzpicture}
\atoms{circ}{0/}
\atoms{triang}{{o0/p={0.8,0},rot=90,lab={t=$O^*$,p=120:0.35}},{o1/p={-0.8,0},rot=-90,lab={t=$O^*$,p=60:0.35}},{o2/p={0,-0.8},lab={t=$O$,p=30:0.35}}}
\draw (0)--(o0-ct) (0)--(o1-ct) (0)--(o2-ct);
\draw[irrep] (o0-mb)--++(0:0.4) (o1-mb)--++(180:0.4) (o2-mb)--++(-90:0.4);
\draw[pind] (o0-cl)--++(0:0.4) (o1-cl)--++(180:0.4) (o2-cl)--++(-90:0.4) (o0-cr)--++(0:0.4) (o1-cr)--++(180:0.4) (o2-cr)--++(-90:0.4);
\end{tikzpicture}
=
\begin{tikzpicture}
\atoms{delta}{0/}
\draw[irrep] (0)--++(0:0.8) (0)--++(180:0.8) (0)--++(-90:0.8);
\draw[pind,rounded corners=5pt] (-0.8,0.2)--(0.8,0.2) (-0.8,-0.2)-|(-0.2,-0.8) (0.2,-0.8)|-(0.8,-0.2);
\end{tikzpicture}\;,\qquad
\begin{tikzpicture}
\atoms{square}{0/}
\atoms{triang}{{o0/p={0.8,0},rot=90,lab={t=$S$,p=-120:0.35}},{o1/p={-0.8,0},rot=-90,lab={t=$S^*$,p=-60:0.35}},{o2/p={0,0.8},rot=180,lab={t=$O^*$,p=-30:0.35}}}
\draw (0)--(o2-ct);
\draw[nqubit] (0)--(o0-ct) (0)--(o1-ct);
\draw[irrep] (o0-mb)--++(0:0.4) (o1-mb)--++(180:0.4) (o2-mb)--++(90:0.4);
\draw[pind] (o0-cr)--++(0:0.4) (o2-cl)--++(90:0.4) (o1-cl)--++(180:0.4) (o2-cr)--++(90:0.4);
\draw[qind] (o0-cl)--++(0:0.4) (o1-cr)--++(180:0.4);
\end{tikzpicture}
=
\begin{tikzpicture}
\atoms{delta}{0/}
\draw[irrep] (0)--++(0:0.8) (0)--++(180:0.8) (0)--++(90:0.8);
\draw[pind,rc] (-0.8,0.2)-|(-0.2,0.8) (0.2,0.8)|-(0.8,0.2);
\draw[qind] (-0.8,-0.2)--(0.8,-0.2);
\end{tikzpicture}\;.
\end{equation}
The small black dot on both right-hand sides is a $\delta$-tensor, which is $1$ if the irrep labels at all indices are equal, and $0$ otherwise. The free lines represent identity matrices. The three identity matrices on the right-hand side of the left equation are matrix multiplication: Each matrix is represented by a double-index, and the second index of the first matrix is contracted with the first index of the second matrix.

We now use the isomorphism $O$ of the symmetry group algebra, and the isomorphism $S$ of its representation. Let us quickly recall what this means for the symmetry group $S_n$ and its qubit permutation representation $R$, even though the discussion here in principle works for general symmetries. As described in the first appendix, the irreps $\lambda$ are Young diagrams and the internal indices label Young tableaux $p_\lambda$. The dimension of the multiplicity index $q$ is only non-zero for Young diagrams with one or two rows. In this case, we can label the irrep by a half-integer $0\leq\lambda\leq\frac{n}{2}$ corresponding to the Young diagram with two rows of lengths $(n/2+\lambda, n/2-\lambda)$, but also to a $SU(2)$ representation. We now let $\widetilde H$ denote the Hamiltonian (or an arbitrary invariant operator) $H$ conjugated with $S$,
\begin{equation}
\label{eq:hamiltonian_block_diagonal}
\begin{tikzpicture}
\atoms{square,yscale={0.4/0.36},lab={t=$\widetilde H$,p=-90:0.5}}{h/}
\draw[irrep] (h-r)--++(0:0.4) (h-l)--++(180:0.4);
\draw[pind] (h-tr)--++(0:0.4) (h-tl)--++(180:0.4);
\draw[qind] (h-br)--++(0:0.4) (h-bl)--++(180:0.4);
\end{tikzpicture}
\coloneqq
\begin{tikzpicture}
\atoms{square,lab={t=$H$,p=-90:0.4}}{h/}
\atoms{triang}{{o0/p={0.8,0},rot=90,lab={t=$S$,p=-120:0.35}},{o1/p={-0.8,0},rot=-90,lab={t=$S^*$,p=-60:0.35}}}
\draw[nqubit] (h-r)--(o0-ct) (h-l)edge[mark={arr,f,s}](o1-ct);
\draw[irrep] (o0-mb)--++(0:0.4) (o1-mb)--++(180:0.4);
\draw[pind] (o0-cr)--++(0:0.4) (o1-cl)--++(180:0.4);
\draw[qind] (o0-cl)--++(0:0.4) (o1-cr)--++(180:0.4);
\end{tikzpicture}\;.
\end{equation}
Eq.~\eqref{eq:h_invariant} in this block-diagonal basis then becomes
\begin{equation}
\label{eq:block_diagonal_invariance}
\begin{tikzpicture}
\atoms{square,yscale={0.4/0.36},lab={t=$\widetilde H$,p=-90:0.5}}{h/}
\atoms{delta}{d/p={-1,0}}
\draw[irrep] (h-r)--++(0:0.4) (h-l)--(d)--++(180:0.6) (d)--++(90:0.6);
\draw[pind,rounded corners=5pt] (h-tr)--++(0:0.4) (h-tl)-|(-0.8,0.6) (-1.2,0.6)|-(-1.6,0.2);
\draw[qind] (h-br)--++(0:0.4) (h-bl)--(-1.6,-0.2);
\end{tikzpicture}
=
\begin{tikzpicture}
\atoms{square,yscale={0.4/0.36},lab={t=$\widetilde H$,p=-90:0.5}}{h/}
\atoms{delta}{d/p={1,0}}
\draw[irrep] (h-l)--++(180:0.4) (h-r)--(d)--++(0:0.6) (d)--++(90:0.6);
\draw[pind,rounded corners=5pt] (h-tl)--++(180:0.4) (h-tr)-|(0.8,0.6) (1.2,0.6)|-(1.6,0.2);
\draw[qind] (h-bl)--++(180:0.4) (h-br)--(1.6,-0.2);
\end{tikzpicture}\;.
\end{equation}
By contracting the left two open internal indices on the left-hand side with copies of a normalized vector $v$ and applying a $\delta$-tensor to the left two irrep indices, we find
\begin{equation}
\label{eq:manybody_h_from_bdiag}
\begin{tikzpicture}
\atoms{square,yscale={0.4/0.36},lab={t=$\widetilde H$,p=-90:0.5}}{h/}
\draw[irrep] (h-r)--++(0:0.4) (h-l)--++(180:0.4);
\draw[pind] (h-tr)--++(0:0.4) (h-tl)--++(180:0.4);
\draw[qind] (h-br)--++(0:0.4) (h-bl)--++(180:0.4);
\end{tikzpicture}
=
\begin{tikzpicture}
\atoms{square,yscale={0.4/0.36},lab={t=$\widetilde H$,p=-90:0.5}}{h/}
\atoms{delta}{d/p={-1,0},d1/p={-2.1,0.8}}
\atoms{circ,small}{{v0/p={-1.2,0.6},lab={t=$v$,p=180:0.25}},{v1/p={-1.8,0.2},lab={t=$v^*$,p=90:0.25}}}
\draw[irrep,rc] (h-r)--++(0:0.4) (h-l)--(d) (d)|-(d1) (d)-|(d1) (d1)--++(180:0.4);
\draw[pind,rc] (h-tr)--++(0:0.4) (h-tl)-|(-0.8,0.6) (v0)|-(v1);
\draw[qind] (h-br)--++(0:0.4) (h-bl)--(-1.6,-0.2);
\end{tikzpicture}
\overset{\text{Eq.~\eqref{eq:block_diagonal_invariance}}}{=}
\begin{tikzpicture}
\atoms{square,yscale={0.4/0.36},lab={t=$\widetilde H$,p=-90:0.5}}{h/}
\atoms{delta}{d/p={0,0.5},d1/p={0,0.8}}
\atoms{circ,small}{{v0/p={-0.75,0.2},lab={t=$v^*$,p=180:0.3}},{v1/p={0.75,0.2},lab={t=$v$,p=0:0.25}}}
\draw[irrep,rc] (h-r)--++(0.25,0)|-(d) (h-l)--++(-0.25,0)|-(d) (d)--(d1) (d1)--++(180:0.8)(d1)--++(0:0.8);
\draw[pind] (h-tl)--(v0) (h-tr)--(v1) (-0.8,1)--(0.8,1);
\draw[qind] (h-br)--++(0:0.6) (h-bl)--++(180:0.6);
\end{tikzpicture}
=
\begin{tikzpicture}
\atoms{square,lab={t=$\widetilde h$,p=-90:0.5}}{h/}
\atoms{delta}{d1/p={0,0.4}}
\draw[irrep] (h-t)--(d1) (d1)--++(180:0.8)(d1)--++(0:0.8);
\draw[pind] (-0.8,0.6)--(0.8,0.6);
\draw[qind] (h-r)--++(0:0.6) (h-l)--++(180:0.6);
\end{tikzpicture}\;,
\end{equation}
where
\begin{equation}
\begin{tikzpicture}
\atoms{square,lab={t=$\widetilde h$,p=-90:0.5}}{h/}
\draw[irrep] (h-t)--++(0,0.4);
\draw[qind] (h-r)--++(0:0.4) (h-l)--++(180:0.4);
\end{tikzpicture}
\coloneqq
\begin{tikzpicture}
\atoms{square,yscale={0.4/0.36},lab={t=$\widetilde H$,p=-90:0.5}}{h/}
\atoms{delta}{d/p={0,0.5}}
\atoms{circ,small}{{v0/p={-0.75,0.2},lab={t=$v^*$,p=180:0.3}},{v1/p={0.75,0.2},lab={t=$v$,p=0:0.25}}}
\draw[irrep,rc] (h-r)--++(0.25,0)|-(d) (h-l)--++(-0.25,0)|-(d) (d)--++(0,0.4);
\draw[pind] (h-tl)--(v0) (h-tr)--(v1);
\draw[qind] (h-br)--++(0:0.6) (h-bl)--++(180:0.6);
\end{tikzpicture}\;.
\end{equation}
Note that $v$ is actually a collection of vectors, depending on the value $\lambda$ of the irrep. Plugging in Eq.~\eqref{eq:hamiltonian_block_diagonal} and Eq.~\eqref{eq:H_from_H}, we find
\begin{equation}
\label{eq:htilde_definition}
\begin{tikzpicture}
\atoms{square,lab={t=$\widetilde h$,p=-90:0.5}}{h/}
\draw[irrep] (h-t)--++(0,0.4);
\draw[qind] (h-r)--++(0:0.4) (h-l)--++(180:0.4);
\end{tikzpicture}
=
\begin{tikzpicture}
\atoms{square,lab={t=$h$,p=-90:0.35}}{h/}
\atoms{triang,rot=180,lab={t=$F$,p=-30:0.3}}{f/p={0,0.6}}
\draw[irrep] (f-mb)--++(90:0.4);
\draw[qind] (f-cl)--++(30:0.4) (f-cr)--++(150:0.4);
\draw (h-t)--(f-ct);
\end{tikzpicture}
\end{equation}
with
\begin{equation}
\label{eq:f_from_a}
\begin{tikzpicture}
\atoms{triang,rot=180,lab={t=$F$,p=-30:0.3}}{f/p={0,0.6}}
\draw[irrep] (f-mb)--++(90:0.4);
\draw[qind] (f-cl)--++(30:0.4) (f-cr)--++(150:0.4);
\draw (f-ct)--++(-90:0.4);
\end{tikzpicture}
\coloneqq
\begin{tikzpicture}
\atoms{square,lab={t=$A$,p=-55:0.35}}{a/}
\atoms{triang}{{s0/rot=90,lab={t=$S$,p=120:0.3},p={0.7,0}},{s1/rot=-90,lab={t=$S^*$,p=60:0.35},p={-0.7,0}}}
\atoms{delta}{d/p={0,0.5}}
\atoms{circ,small}{{v0/p={-1.6,0.2},lab={t=$v^*$,p=180:0.3}},{v1/p={1.6,0.2},lab={t=$v$,p=0:0.25}}}
\draw[irrep,rc] (s0-mb)--++(0.25,0)|-(d) (s1-mb)--++(-0.25,0)|-(d) (d)--++(0,0.4);
\draw[pind] (s1-cl)--(v0) (s0-cr)--(v1);
\draw[qind] (s0-cl)--++(0:0.6) (s1-cr)--++(180:0.6);
\draw[nqubit] (a-l)edge[mark={arr,s}](s1-ct) (a-r)--(s0-ct);
\draw (a-b)--++(-90:0.4);
\end{tikzpicture}\;.
\end{equation}
$F$ is in fact the isomorphism $O$ for the algebra $X$. The tensor $\widetilde h$ can be interpreted as a collection of matrices $\widetilde h^\lambda$ for different irreps $\lambda$. As can be easily seen from Eq.~\eqref{eq:manybody_h_from_bdiag}, the ground state energy of $H$ is the minimum over $\lambda$ of the ground state energies of $\widetilde h^\lambda$. Furthermore, the ground states are given by
\begin{equation}
\label{eq:ground_state_appendix}
\begin{tikzpicture}
\atoms{square}{{p/lab={t=$\psi$,p=-90:0.4}}}
\draw[nqubit] (p-r)--++(0:0.3);
\end{tikzpicture}
=
\begin{tikzpicture}
\atoms{triang}{{s/rot=-90,lab={t=$S$,p=60:0.3}}}
\atoms{circ,small}{{v/p={-0.6,0.3},lab={t=$v$,p=180:0.3}}}
\atoms{circ,small}{{l/p={-0.6,0},lab={t=$\lambda_{\text{min}}$,p=180:0.5}}}
\atoms{circ,small}{{sv/p={-0.6,-0.3},lab={t=$s_{\text{min}}$,p=180:0.5}}}
\draw[irrep] (s-mb)--(l);
\draw[pind] (s-cl)--(v);
\draw[qind] (s-cr)--(sv);
\draw[nqubit] (s-ct)--++(0:0.4);
\end{tikzpicture}\;,
\end{equation}
where $s_{\text{min}}$ is the ground state of $\widetilde h^{\lambda_{\text{min}}}$, and $v$ is an arbitrary vector. Of course, there cannot be an efficient classical description which works for all ground states, since the dimension of $v$ (and therefore the ground state degeneracy) is exponentially large in $n$. However, as we will later give a low-bond dimension MPS description for a full basis of ground states in Eq.~\eqref{eq:ground_state_mps}.

Finally, let us consider the time evolution of a pure state $\ket\psi$ under an invariant unitary $U$. More precisely, we want to calculate the time-evolved expectation value $\bra\psi UOU^\dagger\ket\psi$ of an invariant operator $O$. We find
\begin{equation}
\label{eq:time_evolution_blockdiagonal}
\bra\psi UOU^\dagger\ket\psi
=
\begin{tikzpicture}
\atoms{square}{{p1/p={-1.4,0},lab={t=$\psi^*$,p=-90:0.4}}, {p2/p={1.4,0},lab={t=$\psi$,p=-90:0.4}}}
\atoms{square}{{u1/p={-0.7,0},lab={t=$U$,p=-90:0.4}}, {u2/p={0.7,0},lab={t=$U^*$,p=-90:0.4}}}
\atoms{square,lab={t=$O$,p=-90:0.4}}{o/}
\draw[nqubit] (p1)--(u1)--(o)--(u2)--(p2);
\end{tikzpicture}
=
\begin{tikzpicture}
\atoms{square}{{p1/p={-2.1,0},lab={t=$\psi^*$,p=-90:0.4}}, {p2/p={2.1,0},lab={t=$\psi$,p=-90:0.4}}}
\atoms{square}{{u1/p={-0.7,1.1},lab={t=$u$,p=90:0.3}}, {u2/p={0.7,1.1},lab={t=$u^*$,p=90:0.35}}}
\atoms{square,lab={t=$o$,p=90:0.3}}{o/p={0,1.1}}
\atoms{triang}{{f1/p={-0.7,0.4},lab={t=$F$,p=30:0.3}},{f2/p={0,0.4},lab={t=$F$,p=30:0.3}},{f3/p={0.7,0.4},lab={t=$F^*$,p=30:0.35}}}
\atoms{delta}{d1/p={-0.7,0},d2/p={0,0},d3/p={0.7,0}}
\atoms{triang}{{s0/rot=90,lab={t=$S^*$,p=120:0.3},p={-1.4,0}},{s1/rot=-90,lab={t=$S$,p=60:0.3},p={1.4,0}}}
\draw[nqubit] (p1)--(s0-ct) (p2)--(s1-ct);
\draw[irrep] (s0-mb)--(d1)--(d2)--(d3)--(s1-mb) (d1)--(f1-mb) (d2)--(f2-mb) (d3)--(f3-mb);
\draw[pind] (s0-cl)--(s1-cr);
\draw[qind] (s0-cr)--(f1-cl) (f1-cr)--(f2-cl) (f2-cr)--(f3-cl) (f3-cr)--(s1-cl);
\draw (u1-b)--(f1-ct) (u2-b)--(f3-ct) (o-b)--(f2-ct);
\end{tikzpicture}\;.
\end{equation}
As usual, we cannot efficiently contract this tensor network, since the internal $p$-index and the $n$-qubit indices have an exponential bond dimension in $n$. However, once we have the reduced density matrix
\begin{equation}
\label{eq:reduced_matrix}
\begin{tikzpicture}
\atoms{square,lab={t=$\widetilde\rho$,p=-90:0.4}}{0/}
\draw[irrep] (0-t)edge[ind=$\lambda$]++(90:0.5);
\draw[qind] (0-tl)edge[ind=$q'$]++(120:0.5) (0-tr)edge[ind=$q$]++(60:0.5);
\end{tikzpicture}
\coloneqq
\begin{tikzpicture}
\atoms{square}{{p1/p={-1.5,0},lab={t=$\psi^*$,p=-90:0.4}}, {p2/p={1.5,0},lab={t=$\psi$,p=-90:0.4}}}
\atoms{delta}{d/}
\atoms{triang}{{s0/rot=90,lab={t=$S^*$,p=120:0.3},p={-0.8,0}},{s1/rot=-90,lab={t=$S$,p=60:0.3},p={0.8,0}}}
\draw[nqubit] (p1)--(s0-ct) (p2)--(s1-ct);
\draw[irrep] (s0-mb)--(d)--(s1-mb) (d)edge[ind=$\lambda$]++(90:0.8);
\draw[pind] (s0-cl)--(s1-cr);
\draw[qind,rc,ind=$q'$] (s0-cr)-|(-0.2,0.8);
\draw[qind,rc,ind=$q$] (s1-cl)-|(0.2,0.8);
\end{tikzpicture}\;,
\end{equation}
the remaining contraction is fast, dominated by the contraction between $F$ and $u$ or $o$ which takes time $\operatorname{O}(n^6)$. If $\psi$ is given as many copies of the quantum state, then $\widetilde\rho$ can be obtained by tomography on the $\lambda$ and $q$-subspace, as described in the main text. If $\psi$ is given classically as an MPS, then we show later around Eq.~\eqref{eq:mps_time_evolution} how time evolution can be performed efficiently.

Note that time evolution can also be performed if $U$ is defined through an invariant Hamiltonian $H=A(h)$ via $U=e^{itH}$. In this case, we calculate $\widetilde h$ as in Eq.~\eqref{eq:htilde_definition}, exponentiate the individual $\widetilde h^\lambda$, and plug $e^{it\widetilde h}$ into Eq.~\eqref{eq:time_evolution_blockdiagonal} instead of $u$ and $F$.

\section{Tensor-network method for calculation of the entries of $F$}
\label{sec:tensor_network_f_entries}
In this Appendix, we show how to efficiently calculate the explicit matrix $F$ defined in Eq.~\eqref{eq:f_from_a}. We cannot directly efficiently contract the tensor network on the right-hand side since the two indices shared between $A$ and $S$ have bond dimension $2^n$. In order to make contraction efficient, we will write $A$ and $S$ as MPOs and then contract the tensor network horizontally from qubit to qubit.
Let us start by writing $A$ as an MPO,
\begin{equation}
\label{eq:a_mpo}
\begin{tikzpicture}
\atoms{square}{{0/lab={t=$A$,p=-90:0.35}}}
\draw[nqubit] (0)edge[ind=$\vec l$,mark={arr,f,s}]++(-0.8,0) (0)edge[ind=$\vec m$]++(0.8,0);
\draw (0)edge[triple,ind=$\vec i$]++(90:0.7);
\end{tikzpicture}
=
\sqrt{
\frac{i_0!i_x!i_y!i_z!}{n!2^n}
}
\cdot
\begin{tikzpicture}
\atoms{square,xscale=1,yscale=2.5,lab={t=$+$,p={0,0}}}{0/,1/p={0.7,0},2/p={2.1,0}}
\atoms{circ,tiny,lab={t=$0$,p=180:0.25}}{0x/p={-0.5,-0.3},0y/p={-0.5,0},0z/p={-0.5,0.3}}
\draw ([sy=-0.3]2-r)edge[ind=$i_x$]++(0:0.4) (2-r)edge[ind=$i_y$]++(0:0.4) ([sy=0.3]2-r)edge[ind=$i_z$]++(0:0.4);
\draw ([sy=-0.3]1-r)edge[mark={three dots,a}]++(0:0.3) (1-r)edge[mark={three dots,a}]++(0:0.3) ([sy=0.3]1-r)edge[mark={three dots,a}]++(0:0.3);
\draw ([sy=-0.3]2-l)edge[]++(180:0.3) (2-l)edge[]++(180:0.3) ([sy=0.3]2-l)edge[]++(180:0.3);
\draw ([sy=-0.3]0-r)--([sy=-0.3]1-l) (0-r)--(1-l) ([sy=0.3]0-r)--([sy=0.3]1-l);
\draw (0x)--([sy=-0.3]0-l) (0y)--(0-l) (0z)--([sy=0.3]0-l);
\draw (0-t)edge[ind=$l_0$]++(90:0.4) (1-t)edge[ind=$l_1$]++(90:0.4) (2-t)edge[ind=$l_{n-1}$]++(90:0.4);
\draw (0-b)edge[ind=$m_0$]++(-90:0.4) (1-b)edge[ind=$m_1$]++(-90:0.4) (2-b)edge[ind=$m_{n-1}$]++(-90:0.4);
\end{tikzpicture}\;.
\end{equation}
The global prefactor on the right-hand side was not included in the definition of $A$ in the main text, but is necessary for the normalization in Eq.~\eqref{eq:a_normalization}. The arrows over the indices on the left-hand side indicate that they correspond to collections of indices on the right-hand side. We have also used a triple-line for the $X$-index $\vec i$, to indicate that it is labelled by a triple $(i_x,i_y,i_z)$. The bond dimension of all three horizontal indices is $n$, and the tensor with the $+$ label is defined as
\begin{equation}
\begin{tikzpicture}
\atoms{square,xscale=1,yscale=2.5,lab={t=$+$,p={0,0}}}{0/}
\draw ([sy=-0.3]0-r)edge[ind=$j_x$]++(0:0.4) (0-r)edge[ind=$j_y$]++(0:0.4) ([sy=0.3]0-r)edge[ind=$j_z$]++(0:0.4);
\draw ([sy=-0.3]0-l)edge[ind=$i_x$]++(180:0.4) (0-l)edge[ind=$i_y$]++(180:0.4) ([sy=0.3]0-l)edge[ind=$i_z$]++(180:0.4);
\draw (0-t)edge[ind=$l$]++(90:0.4) (0-b)edge[ind=$m$]++(-90:0.4);
\end{tikzpicture}
=
\begin{multlined}
(\operatorname{id}_2)_{lm} \delta_{i_x,j_x} \delta_{i_y,j_y} \delta_{i_z,j_z}
+(\sigma_x)_{lm} \delta_{i_x,j_x+1} \delta_{i_y,j_y} \delta_{i_z,j_z}\\
+(\sigma_y)_{lm} \delta_{i_x,j_x} \delta_{i_y,j_y+1} \delta_{i_z,j_z}
+(\sigma_z)_{lm} \delta_{i_x,j_x} \delta_{i_y,j_y} \delta_{i_z,j_z+1}
\end{multlined}\;.
\end{equation}
Note that the value of $\delta_{a,b}$ is understood to be zero unless $0\leq a,b<n$. We can see that the value of the bottom, middle, and top horizontal index in Eq.~\eqref{eq:a_mpo} is the count of $X$s, $Y$s, and $Z$s up to this point, respectively. To achieve this, the value of the bottom horizontal indices of a $+$ tensor increases by $1$ from left to right if the two vertical indices are in an $X$ configuration, and analogously for the middle and top horizontal indices. 

In order to write $S$ as an MPO, we make use of the fact that the commutant of the representation $R$ of $S_n$ is spanned by the representation $(\frac12)^{\otimes n}$ of $SU(2)$. Thus, the isomorphism $S$ is the same for both those representations, just that the internal and multiplicity indices are exchanged. Now, the representation $(\frac12)^{\otimes n}$ is a tensor product of on-site representations, and thus an MPO,
\begin{equation}
\label{eq:onsite_representation}
\begin{tikzpicture}
\atoms{square}{{0/lab={t=$(\frac12)^{\otimes n}$,p=-90:0.45}}}
\draw[nqubit] (0)edge[ind=$\vec l$]++(-0.8,0) (0)edge[ind=$\vec m$]++(0.8,0);
\draw (0)edge[ind=$u$]++(90:0.5);
\end{tikzpicture}
=
\begin{tikzpicture}
\atoms{square}{{r0/lab={t=$r_0$,p=180:0.35}},{r1/p={1,0},lab={t=$r_1$,p=180:0.35}},{r2/p={3,0},lab={t=$r_{n-1}$,p=180:0.55}}}
\atoms{delta}{d0/p={1.4,0.4}, d1/p={3.4,0.4}}
\draw[rc] (r0-r)-|(0.4,0.4)--(d0) (d0)edge[mark={three dots,a}]++(0:0.3) (d1)edge[mark={three dots,a}]++(180:0.6) (d1)edge[ind=$u$]++(0.5,0) (r1)-|(d0) (r2)-|(d1);
\draw (r0-t)edge[ind=$l_0$]++(90:0.5) (r1-t)edge[ind=$l_1$]++(90:0.5) (r2-t)edge[ind=$l_{n-1}$]++(90:0.5) (r0-b)edge[ind=$m_0$]++(-90:0.3) (r1-b)edge[ind=$m_1$]++(-90:0.3) (r2-b)edge[ind=$m_{n-1}$]++(-90:0.3);
\end{tikzpicture}
\;.
\end{equation}
Note that $u$ labels an element of $SU(2)$, so this is a tensor network with a continuous bond dimension, which might not be directly suitable for practical computation but still makes sense formally. In our case all the $r_i$ are the spin-$\frac12$ representation, but we will discuss the following few steps for arbitrary representations. In order to get the isomorphism $S$ for the overall representation (for us $(\frac12)^{\otimes n}$), we start by block-diagonalizing the individual $r_i$. We end up with many copies of the isomorphism $O$ for the group algebra of $SU(2)$. Those copies of $O$ need to be pulled through all the $\delta$-tensors. When doing this, the $\delta$-tensor splits up into two copies of a tensor $C$,
\begin{equation}
\begin{tikzpicture}
\atoms{delta}{0/}
\atoms{triang}{{o0/p=-45:0.8,rot=45,lab={t=$O$,p=90:0.35}},{o1/p=-135:0.8,rot=-45,lab={t=$O$,p=90:0.35}},{o2/p={0,0.8},rot=180,lab={t=$O^*$,p=-30:0.35}}}
\draw (0)--(o0-ct) (0)--(o1-ct) (0)--(o2-ct);
\draw[irrep] (o0-mb)--++(-45:0.4) (o1-mb)--++(-135:0.4) (o2-mb)--++(90:0.4);
\draw[qind] (o0-cl)--++(-45:0.4) (o1-cl)--++(-135:0.4) (o2-cl)--++(90:0.4) (o0-cr)--++(-45:0.4) (o1-cr)--++(-135:0.4) (o2-cr)--++(90:0.4);
\end{tikzpicture}
=
\begin{tikzpicture}
\atoms{circ,dot,big}{{0/p={-0.5,0},lab={t=$C^*$,p=180:0.5}},{1/p={0.5,0},lab={t=$C$,p=0:0.45}}}
\atoms{delta}{d0/p={-0.5,-1}, d1/p={0.5,-1}, d2/p={0,0.5}}
\draw[irrep] (0)--(d0) (0)--(d1) (0)--(d2) (1)--(d0) (1)--(d1) (1)--(d2) (d0)--++(-90:0.4) (d1)--++(-90:0.4) (d2)--++(90:0.4);
\draw[rc,qind] (0.-120)--++(-90:1.2) (1.-60)--++(-90:1.2) (0.-75)--++(-45:1.1)--++(-90:0.35) (1.-105)--++(-135:1.1)--++(-90:0.35) (0.75)--++(45:0.4)--++(90:0.35) (1.105)--++(135:0.4)--++(90:0.35);
\draw[fusion] (0)--(1);
\end{tikzpicture}\;.
\end{equation}
$C$ exists since the group algebra together with the $\delta$-tensor forms a Hopf algebra. When writing the bi-algebra axiom in the block-diagonal basis, one can see that the left-hand side for a fixed configuration of irrep indices is a projector $P$ between all three left internal indices to all three right internal indices. So $C$ can be defined as an isometry such that $CC^\dagger=P$. The index connecting the two copies of $C$ has a bond dimension which depends on the three irreps, and will be referred to as \emph{fusion index}. The entries of $C$ are known as the \emph{Clebsch-Gordon coefficients}. We can now use this equation to pull the isomorphism $O$ through all the $\delta$-tensors in Eq.~\eqref{eq:onsite_representation}. As a result, we obtain $S$ as
\begin{equation}
\begin{tikzpicture}
\atoms{triang,rot=90,lab={t=$S$,p=-120:0.35}}{0/}
\draw[irrep] (0-mb)edge[ind=$\lambda$]++(0:0.5);
\draw[nqubit] (0-ct)edge[ind=$\vec l$]++(180:0.5);
\draw[qind] (0-cr)edge[ind=$q$]++(10:0.5);
\draw (0-cl)edge[pind,postaction={ind=$\vec p$}]++(-10:0.5);
\end{tikzpicture}
=
\begin{tikzpicture}
\atoms{circ,dot,big}{0/, 1/p={1.5,0}, 2/p={4,0}}
\atoms{delta}{d0/p={0.7,0}, d1/p={2.2,0}}
\atoms{triang,rot=180}{{s0/p={0,-0.8},lab={t=$s_0$,p=-30:0.3}}, {s1/p={1.5,-0.8},lab={t=$s_1$,p=-30:0.3}}, {s2/p={4,-0.8},lab={t=$s_{n-1}$,p=-20:0.5}}}
\atoms{circ,tiny,lab={t=$0$,p=180:0.25}}{c0/p={-0.6,0}}
\draw (0)edge[fusion,postaction={ind=$p_0''$}]++(90:0.6) (1)edge[fusion,postaction={ind=$p_1''$}]++(90:0.6) (2)edge[fusion,postaction={ind=$p_{n-1}''$}]++(90:0.8);
\draw[irrep] (0)--(c0) (0)--(d0) (d0)--(1) (1)--(d1) (d1)edge[mark={three dots,a}]++(0:0.3) (2)edge[mark={three dots,a}]++(180:0.4)  (2.30)edge[ind=$\lambda$]++(0:0.7) (d0)edge[ind=$p_0$]++(90:0.5) (d1)edge[ind=$p_1$]++(90:0.5) (s0-mb)--(0) (s1-mb)--(1) (s2-mb)--(2);
\draw[smallmultip,ind=$p_0'$] (s0-cl)--++(60:0.3)--++(90:1);
\draw[smallmultip,ind=$p_1'$] (s1-cl)--++(60:0.3)--++(90:1);
\draw[smallmultip,ind=$p_{n-1}'$] (s2-cl)--++(60:0.5)--++(90:0.6);
\draw[qind] (s0-cr)--++(90:0.2)--(0) (s1-cr)--++(90:0.2)--(1) (s2-cr)--++(90:0.2)--(2) (0.-30)--(1.-150);
\draw[qind] (1.-30)edge[postaction={mark={three dots,a}}]++(0:0.7) (2.-150)edge[postaction={mark={three dots,a}}]++(180:0.4) (2.-30)edge[postaction={ind=$q$}]++(0:0.75);
\draw (s0-ct)edge[ind=$l_0$]++(-90:0.3) (s1-ct)edge[ind=$l_1$]++(-90:0.3) (s2-ct)edge[ind=$l_{n-1}$]++(-90:0.3);
\end{tikzpicture}
\;.
\end{equation}
For our present situation where the group is $SU(2)$, the Clebsch-Gordon coefficients are well-known and can be computed efficiently up to $p$ bits of precision (i.e., up to additive error exponentially small in $p$) in time $\operatorname{poly}\left(n,p\right)$ by the Racah formula~\cite{PhysRev.62.438}. Note that the fusion index is always trivial, that is, it has dimension either $0$ or $1$. All local representations $r_i$ are equal to spin-$\frac12$ such that all the local $s_i$ can be chosen equal to the identity. With those simplifications, we get,
\begin{equation}
\label{eq:s_mpo_simplified}
\begin{tikzpicture}
\atoms{triang,rot=90,lab={t=$S$,p=-120:0.35}}{0/}
\draw[irrep] (0-mb)edge[ind=$\lambda$]++(0:0.5);
\draw[nqubit] (0-ct)edge[ind=$\vec l$]++(180:0.5);
\draw[qind] (0-cr)edge[ind=$q$]++(10:0.5);
\draw (0-cl)edge[pind,postaction={ind=$\vec p$}]++(-10:0.5);
\end{tikzpicture}
=
\begin{tikzpicture}
\atoms{circ,dot,big}{0/, 1/p={1.5,0}, 2/p={4,0}}
\atoms{delta}{d0/p={0.7,0}, d1/p={2.2,0}}
\atoms{circ,tiny,lab={t=$\frac12$,p=0:0.25}}{{s0/p={0,-0.6}},{s1/p={1.5,-0.6}},{s2/p={4,-0.6}}}
\atoms{circ,tiny,lab={t=$0$,p=180:0.25}}{c0/p={-0.6,0}}
\draw[irrep] (0)--(c0) (0)--(d0) (d0)--(1) (1)--(d1) (d1)edge[mark={three dots,a}]++(0:0.3) (2)edge[mark={three dots,a}]++(180:0.4)  (2.30)edge[ind=$\lambda$]++(0:0.7) (d0)edge[ind=$p_0$]++(90:0.5) (d1)edge[ind=$p_1$]++(90:0.5) (s0)--(0) (s1)--(1) (s2)--(2);
\draw[qind] (0.-120)edge[ind=$l_0$]++(-90:0.8) (1.-120)edge[ind=$l_1$]++(-90:0.8) (2.-120)edge[ind=$l_{n-1}$]++(-90:0.8) (0.-30)--(1.-150);
\draw[qind] (1.-30)edge[postaction={mark={three dots,a}}]++(0:0.7) (2.-150)edge[postaction={mark={three dots,a}}]++(180:0.4) (2.-30)edge[postaction={ind=$q$}]++(0:0.75);
\end{tikzpicture}
\;.
\end{equation}
The tensor product of the spin-$\frac{c}{2}$ representation with the spin-$\frac12$ representation is $\frac{c}{2}\otimes \frac12 = \frac{c-1}{2}\oplus\frac{c+1}{2}$. Thus, any two consecutive irreps $p_i,p_{i+1}$ in the sequence $p_0,\ldots,p_{n-1}=\lambda$ must differ by $\pm \frac12$, otherwise the MPO above yields $0$. Such a sequence exists if and only if $(2\lambda)\mod 2=n \mod 2$ and $\lambda\leq \frac{n}{2}$. In fact, one can directly see the correspondence between such sequences $\vec p$ and the Young tableaux for the Young diagram with two rows of lengths $(n/2+\lambda, n/2-\lambda)$: Starting at an empty diagram, we fill its fields consecutively with the numbers $0,\ldots,n-1$. In the $i$th step, the number $i$ is appended to the first row if $p_i=p_{i-1}+\frac12$, and added to the second row if $p_i=p_{i-1}-\frac12$. Note that a representation of the Schur-Weyl basis equivalent to this MPO has also been used in Ref.~\cite{PhysRevLett.121.060505}, and was originally presented in Ref.~\cite{harrowschur}.

Let us now plug the derived MPO representations in Eq.~\eqref{eq:s_mpo_simplified} and Eq.~\eqref{eq:a_mpo} into Eq.~\eqref{eq:f_from_a}. For the vector $v$, we choose a fixed basis vector corresponding to a valid sequence $\vec p=(p_0,\ldots,p_{n-1})$ depending on $\lambda$,
\begin{equation}
\begin{tikzpicture}
\atoms{triang,rot=180,lab={t=$F$,p=-30:0.3}}{f/p={0,0.6}}
\draw[irrep] (f-mb)edge[ind=$\lambda$]++(90:0.4);
\draw[qind] (f-cl)edge[ind=$q$]++(30:0.4) (f-cr)edge[ind=$q'$]++(150:0.4);
\draw (f-ct)edge[triple,ind=$\vec i$]++(-90:0.4);
\end{tikzpicture}
=
\begin{tikzpicture}
\atoms{square,yscale=1.3,lab={t=$+$,p={0,0}}}{p0/p={-0.2,0},p1/p={1.3,0},p2/p={3.8,0}}
\atoms{circ,dot,big,lab={t=$*$,p=60:0.4}}{0/p={0,1}, 1/p={1.5,1}, 2/p={4,1}}
\atoms{circ,dot,big}{0x/p={0,-1}, 1x/p={1.5,-1}, 2x/p={4,-1}}
\atoms{delta}{d0/p={0.7,1}, d1/p={2.2,1}, d0x/p={0.7,-1}, d1x/p={2.2,-1}}
\atoms{circ,tiny}{{pp0/p={0.7,1.4},lab={t=$p_0$,p=90:0.25}}, {pp1/p={2.2,1.4},lab={t=$p_1$,p=90:0.25}}, {pp0x/p={0.7,-1.4},lab={t=$p_0$,p=-90:0.25}}, {pp1x/p={2.2,-1.4},lab={t=$p_1$,p=-90:0.25}}}
\atoms{circ,tiny,lab={t=$\frac12$,p=0:0.25}}{{s0/p={0,0.5}},{s1/p={1.5,0.5}},{s2/p={4,0.5}}}
\atoms{circ,tiny,lab={t=$\frac12$,p=0:0.25}}{{s0x/p={0,-0.5}},{s1x/p={1.5,-0.5}},{s2x/p={4,-0.5}}}
\atoms{circ,tiny,lab={t=$\lambda$,p=0:0.25}}{{l/p={4.6,1.1}},{lx/p={4.6,-1.1}}}
\atoms{circ,tiny,lab={t=$0$,p=180:0.25}}{c0/p={-0.6,1}, c0x/p={-0.6,-1}}
\atoms{circ,tiny,lab={t=$000$,p=180:0.4}}{c00/p={-0.8,0}}
\draw[irrep] (0)--(c0) (0)--(d0) (d0)--(1) (1)--(d1) (d1)edge[mark={three dots,a}]++(0:0.3) (2)edge[mark={three dots,a}]++(180:0.4)  (2)--(l) (d0)--(pp0) (d1)--(pp1) (s0)--(0) (s1)--(1) (s2)--(2);
\draw[qind] (0.-120)--(p0-t) (p0-b)--(0x.120) (1.-120)--(p1-t) (p1-b)--(1x.120) (2.-120)--(p2-t) (p2-b)--(2x.120) (0.-30)--(1.-150) (0x.30)--(1x.150);
\draw[qind] (1.-30)edge[postaction={mark={three dots,a}}]++(0:0.7) (2.-150)edge[postaction={mark={three dots,a}}]++(180:0.4) (2.-30)edge[postaction={ind=$q'$}]++(0:0.75);
\draw[irrep] (0x)--(c0x) (0x)--(d0x) (d0x)--(1x) (1x)--(d1x) (d1x)edge[mark={three dots,a}]++(0:0.3) (2x)edge[mark={three dots,a}]++(180:0.4)  (2x)--(lx) (d0x)--(pp0x) (d1x)--(pp1x) (s0x)--(0x) (s1x)--(1x) (s2x)--(2x);
\draw[qind] (1x.30)edge[postaction={mark={three dots,a}}]++(0:0.7) (2x.150)edge[postaction={mark={three dots,a}}]++(180:0.4) (2x.30)edge[postaction={ind=$q$}]++(0:0.75);
\draw (p0-r)edge[triple](p1-l) (p1-r)edge[triple,mark={three dots,a}]++(0:0.6) (p2-l)edge[triple,mark={three dots,a}]++(180:0.4) (p2-r)edge[triple,ind=$\vec i$]++(0:0.5) (p0-l)edge[triple](c00);
\end{tikzpicture}
\;.
\end{equation}
The maximal possible value of the $p_i$ is $\operatorname{O}(n)$, and thus the bond dimension of the horizontal $q$-indices is $\operatorname{O}(n)$ as well. We can already see that contracting the tensor network from the left to the right has polynomial runtime. The contraction becomes even faster by the following considerations. For a fixed $\lambda'=\lambda\pm \frac12$, consider the Clebsch-Gordon tensor,
\begin{equation}
\begin{tikzpicture}
\atoms{circ,dot,big}{0/}
\atoms{circ,tiny,lab={t=$\frac12$,p=0:0.25}}{s/p={0,-0.5}}
\draw[irrep] (0)--(s) (0)edge[ind=$\lambda$]++(180:0.6) (0)edge[ind=$\lambda'$]++(0:0.6);
\draw[qind] (0.-120)edge[ind=$l$]++(-90:0.6) (0.150)edge[ind=$q$]++(180:0.8) (0.30)edge[ind=$q'$]++(0:0.8);
\end{tikzpicture}\;,
\end{equation}
as collection of matrices $M_{qq'}$ for different values of $l$. Those matrices are simultaneously constant-width block-diagonal. The $+$ tensor is also band-diagonal, and the bond dimension of the vertical indices is $\operatorname{O}(1)$. Since there are $5$ horizontal indices each of bond dimension $\operatorname{O}(n)$ and all tensors are band diagonal, adding one column of tensors to the contraction takes time $\operatorname{O}(n^5)$. Since there are $n$ qubits/contraction steps, as well as $\operatorname{O}(n)$ different values of $\lambda$, the total runtime for the contraction is $\operatorname{O}(n^7)$.

\section{Time evolution and ground state using MPS}
\label{sec:time_evo_appendix}
In this appendix, discuss the two algorithms from Theorems~\ref{thm:ground_state_mps} and~\ref{thm:time_evo_mps}, which are classically end to end using MPS as input or output, in more detail.
Let us start by Theorem~\ref{thm:ground_state_mps} giving an MPS description of the ground states of an invariant Hamiltonian by combining Eq.~\eqref{eq:ground_state_appendix} and Eq.~\eqref{eq:s_mpo_simplified}. A basis of ground states is labelled by valid sequences $p_0,\ldots,p_{n-1}=\lambda_{\text{min}}$. For every such sequence, the ground state is given as an MPS of bond dimension $\operatorname{O}(n)$,
\begin{equation}
\label{eq:ground_state_mps}
\begin{tikzpicture}
\atoms{square}{{p/lab={t=$\psi_{\text{min}}$,p=-90:0.4}}}
\draw[nqubit] (p-r)edge[ind=$\vec l$]++(0:0.4);
\end{tikzpicture}
=
\begin{tikzpicture}
\atoms{circ,dot,big}{0/, 1/p={1.5,0}, 2/p={4,0}}
\atoms{delta}{d0/p={0.7,0}, d1/p={2.2,0}}
\atoms{circ,tiny}{{pp0/p={0.7,0.5},lab={t=$p_0$,p=90:0.25}}, {pp1/p={2.2,0.5},lab={t=$p_1$,p=90:0.25}}}
\atoms{circ,tiny,lab={t=$\frac12$,p=0:0.25}}{{s0/p={0,-0.6}},{s1/p={1.5,-0.6}},{s2/p={4,-0.6}}}
\atoms{circ,tiny,lab={t=$0$,p=180:0.25}}{c0/p={-0.6,0}}
\atoms{circ,tiny,lab={t=$\lambda_{\text{min}}$,p=0:0.5}}{l/p={4.6,0.2}}
\atoms{circ,tiny,lab={t=$s_{\text{min}}$,p=0:0.5}}{s/p={4.6,-0.2}}
\draw[irrep] (0)--(c0) (0)--(d0) (d0)--(1) (1)--(d1) (d1)edge[mark={three dots,a}]++(0:0.3) (2)edge[mark={three dots,a}]++(180:0.4)  (2.30)--(l) (d0)--(pp0) (d1)--(pp1) (s0)--(0) (s1)--(1) (s2)--(2);
\draw[qind] (0.-120)edge[ind=$l_0$]++(-90:0.8) (1.-120)edge[ind=$l_1$]++(-90:0.8) (2.-120)edge[ind=$l_{n-1}$]++(-90:0.8) (0.-30)--(1.-150);
\draw[qind] (1.-30)edge[postaction={mark={three dots,a}}]++(0:0.7) (2.-150)edge[postaction={mark={three dots,a}}]++(180:0.4) (2.-30)--(s);
\end{tikzpicture}
\;.
\end{equation}

Next, let us discuss Theorem~\ref{thm:time_evo_mps}, which calculates the time-evolved expectation value giving an initial state $\psi$ as an MPS,
\begin{equation}
\begin{tikzpicture}
\atoms{square,lab={t=$\psi$,p=-90:0.3}}{0/}
\draw[nqubit] (0-r)edge[ind=$\vec l$]++(0:0.4);
\end{tikzpicture}
=
\begin{tikzpicture}
\atoms{square,small}{{a0/p={-0.2,1.4},lab={t=$A_0$,p=-90:0.3}}, {a1/p={1.3,1.4},lab={t=$A_1$,p=-90:0.3}}, {a2/p={3.8,1.4},lab={t=$A_{n-1}$,p=-90:0.3}}}
\draw[mps] (a0)--(a1) (a1)edge[mark={three dots,a}]++(0:0.9) (a2)edge[mark={three dots,a}]++(180:0.5);
\draw (a0-t)edge[ind=$l_0$]++(90:0.4) (a1-t)edge[ind=$l_1$]++(90:0.4) (a2-t)edge[ind=$l_{n-1}$]++(90:0.4);
\end{tikzpicture}\;,
\end{equation}
of bond dimension $\chi$. We first note that for $u$ representing an invariant unitary $U=A(u)$, $U\ket\psi$ is again an MPS of bond dimension $\chi \operatorname{O}(n^3)$ due to the MPO representation of $A$ in Eq.~\eqref{eq:a_mpo}. An invariant observable $O=A(o)$ can be written as an MPO as well, so the expectation value $\bra\psi UOU^\dagger\ket\psi$ can be evaluated as a product of MPS and MPO from left to right in polynomial time. However, the most efficient way to calculate that expectation value is to build on the method presented around Eq.~\eqref{eq:reduced_matrix}. The reduced density matrix $\widetilde\rho$ in Eq.~\eqref{eq:reduced_matrix} can be computed efficiently using the MPO representation of $S$ in Eq.~\eqref{eq:s_mpo_simplified},
\begin{equation}
\label{eq:mps_time_evolution}
\begin{tikzpicture}
\atoms{square,lab={t=$\widetilde\rho$,p=-90:0.4}}{0/}
\draw[irrep] (0-t)edge[ind=$\lambda$]++(90:0.5);
\draw[qind] (0-tl)edge[ind=$q'$]++(120:0.5) (0-tr)edge[ind=$q$]++(60:0.5);
\end{tikzpicture}
=
\begin{tikzpicture}
\atoms{circ,dot,big,lab={t=$*$,p=-60:0.4}}{0/p={0,0.5}, 1/p={1.5,0.5}, 2/p={4,0.5}}
\atoms{circ,dot,big}{0x/p={0,-0.5}, 1x/p={1.5,-0.5}, 2x/p={4,-0.5}}
\atoms{delta}{d0/p={0.7,0.5}, d1/p={2.2,0.5}, d0x/p={0.7,-0.5}, d1x/p={2.2,-0.5}, dm/p={4.8,0}}
\atoms{square,small}{{a0/p={-0.2,1.4},lab={t=$A_0^*$,p=90:0.3}}, {a1/p={1.3,1.4},lab={t=$A_1^*$,p=90:0.3}}, {a2/p={3.8,1.4},lab={t=$A_{n-1}^*$,p=90:0.3}}}
\atoms{square,small}{{a0x/p={-0.2,-1.4},lab={t=$A_0$,p=-90:0.3}}, {a1x/p={1.3,-1.4},lab={t=$A_1$,p=-90:0.3}}, {a2x/p={3.8,-1.4},lab={t=$A_{n-1}$,p=-90:0.3}}}
\atoms{circ,tiny,lab={t=$\frac12$,p=0:0.25}}{{s0/p={0,1}},{s1/p={1.5,1}},{s2/p={4,1}}}
\atoms{circ,tiny,lab={t=$\frac12$,p=0:0.25}}{{s0x/p={0,-1}},{s1x/p={1.5,-1}},{s2x/p={4,-1}}}
\atoms{circ,tiny,lab={t=$0$,p=180:0.25}}{c0/p={-0.6,0.5}, c0x/p={-0.6,-0.5}}
\draw[irrep] (0)--(c0) (0)--(d0) (d0)--(1) (1)--(d1) (d1)edge[mark={three dots,a}]++(0:0.3) (2)edge[mark={three dots,a}]++(180:0.4)  (d0)--(d0x) (d1)--(d1x) (s0)--(0) (s1)--(1) (s2)--(2);
\draw[qind] (0.120)--(a0) (a0x)--(0x.-120) (1.120)--(a1) (a1x)--(1x.-120) (2.120)--(a2) (a2x)--(2x.-120) (0.30)--(1.150) (0x.-30)--(1x.-150);
\draw[qind] (1.30)edge[postaction={mark={three dots,a}}]++(0:0.7) (2.150)edge[postaction={mark={three dots,a}}]++(180:0.4) (2.30)edge[postaction={ind=$q'$}]++(0:0.75);
\draw[irrep] (0x)--(c0x) (0x)--(d0x) (d0x)--(1x) (1x)--(d1x) (d1x)edge[mark={three dots,a}]++(0:0.3) (2x)edge[mark={three dots,a}]++(180:0.4) (s0x)--(0x) (s1x)--(1x) (s2x)--(2x);
\draw[qind] (1x.-30)edge[postaction={mark={three dots,a}}]++(0:0.7) (2x.-150)edge[postaction={mark={three dots,a}}]++(180:0.4) (2x.-30)edge[postaction={ind=$q$}]++(0:0.75);
\draw[irrep,rc] (2)--++(0:0.5)--(dm) (2x)--++(0:0.5)--(dm) (dm)edge[ind=$\lambda$]++(0.4,0);
\draw[mps] (a0)--(a1) (a1)edge[mark={three dots,a}]++(0:0.9) (a2)edge[mark={three dots,a}]++(180:0.5);
\draw[mps] (a0x)--(a1x) (a1x)edge[mark={three dots,a}]++(0:0.9) (a2x)edge[mark={three dots,a}]++(180:0.5);
\end{tikzpicture}\;.
\end{equation}
As usual, this tensor network can be contracted from the left to the right. There are five independent horizontal indices, two of them of bond dimension $\chi$, and three of bond dimension $\operatorname{O}(n)$. As argued above, the Clebsch-Gordon coefficients are band-diagonal. This is even true if the irrep indices are not fixed, since they can only change by $\pm\frac12$. The runtime of adding the $i$th column of tensors to the contraction is thus dominated by contracting the MPS tensors $A_i$, which takes time $\operatorname{O}(\chi^w n^3)$. Since there are $n$ steps, the total runtime is $\operatorname{O}(\chi^w n^4)$.

\section{Alternative combinatorial method for calculation of the entries of $F$}
\label{sec:combinatorial_f}
In Appendix~\ref{sec:tensor_network_f_entries}, we have presented a tensor-network method for calculating the matrix entries $F_{q_\lambda,q_\lambda'}^{\bm{i},\lambda}$ used in the main text to efficiently simulate permutation symmetric Hamiltonians. In this section, we present an alternative method that calculated these matrix elements in polynomial runtime $\operatorname{O}(n^{10})$ using combinatorics. Even though the runtime scaling for this method is slower, it may be more approachable and insightful for people without tensor-network background.

We also like to remark that in this method, the matrix entries are calculated individually, each taking a runtime of $\operatorname{O}(n^4)$. This might help if the given Hamiltonian is sparse in the symmetrized Pauli basis. In this case it suffices to calculate only the matrix elements for the non-zero Pauli monomials. For example, consider a family of permutation invariant Hamiltonians that is $k$-local for a constant $k$. In this case, we only need the matrix elements for a constant number of Pauli monomials and thus in this case (since the Schur basis has dimension $\operatorname{O}(n^3)$) the relevant matrix elements can also be calculated in $\operatorname{O}(n^7)$. In contrast, restricting to a sparse subset of symmetrized Pauli monomials does not directly lead to an improvment of the $\operatorname{O}(n^7)$ runtime of the tensor-network based algorithm.

The combinatorial formula for the matrix elements that we will derive in this appendix is as follows.
\begin{lemma}
Recalling that $\lambda$ is given by a Young diagram, we choose $p_{\lambda0}$ to be the standard Young tableaux for that diagram, with numbers increasing first in the column direction and then in row direction, as shown in Eq.~\eqref{eq:standard_young_tableaux}. Then the tensor components $F_{q_\lambda,q_\lambda'}^{\bm{i},\lambda}$ discussed in the main text for the completely symmetrized Pauli representation are given by:
\begin{equation}
\label{eq:sn_f_formula}
\begin{multlined}
    F_{q_\lambda,q_\lambda'}^{\bm{i},\lambda}=\sum_{\substack{f_{11},f_{xx},f_{yy},f_{zz},\\g_{010},g_{111},g_{0x1},g_{1x0},\\g_{0y1},g_{1y0},g_{0z0},g_{1z1}\\/Eq.~\eqref{eq:f_formula_constraints}}}
    \frac{1}{\sqrt{\binom{n-2\lambda_1}{q_\lambda}{\binom{n-2\lambda_1}{q_\lambda'}}}}
    i^{2f_{xx}+2f_{yy}+2f_{zz}+2g_{1z1}-g_{0y1}+g_{1y0}}\\
    \cdot\frac{\lambda_1!(n-2\lambda_1)!}{f_{11}!f_{xx}!f_{yy}!f_{zz}!g_{010}!g_{111}!g_{0x1}!g_{1x0}!g_{0y1}!g_{1y0}!g_{0z0}!g_{1z1}!}\;,
\end{multlined}
\end{equation}
where the sum is over a set of 12 non-negative integers fulfilling the constraints
\begin{equation}
\label{eq:f_formula_constraints}
    \begin{gathered}
    g_{010}+g_{0z0}+g_{0x1}+g_{0y1} = n-2\lambda_1-q_\lambda,\\
    g_{010}+g_{0z0}+g_{1x0}+g_{1y0} = n-2\lambda_1-q_\lambda',\\
    g_{111}+g_{1z1}+g_{1x0}+g_{1y0} = q_\lambda,\\
    g_{111}+g_{1z1}+g_{0x1}+g_{0y1} = q_\lambda',\\
    2f_{11}+g_{010}+g_{111} = i_1,\\
    2f_{xx}+g_{0x1}+g_{1x0} = i_x,\\
    2f_{yy}+g_{0y1}+g_{1y0} = i_y,\\
    2f_{zz}+g_{0z0}+g_{1z1} = i_z
    \end{gathered}
\end{equation}
and $\lambda_1$ is the length of the second row of $\lambda$.
\end{lemma}
\begin{proof}
Following Appendix~\ref{app:schur_basis}, we can project onto the space with an $\mathrm{S}_n$ irrep $\lambda$ and a fixed multiplicity label $p_{\lambda0}$ using the Young symmetrizer in Eq.~\eqref{eq:young_symmetrizer}. Acting with the Young symmetrizer on a computational basis state yields a superposition of basis states with the same number of $0$s and $1$s. Let us write $\lambda = (\lambda_0, \lambda_1)$ for the lengths of the first and second row of $\lambda$. Then we see that applying the Young symmetrizer yields $0$ unless the number of $1$s is between $\lambda_1$ and $\lambda_0$. This is because the row symmetrizer does not change the number of $1$s, and the antisymmetrizer on $\lambda_1$ length-$2$ columns yields $0$ if any columns are $00$ or $11$. Thus, the irrep basis states can be obtained by applying the Young symmetrizer to states with $\lambda_1+q_\lambda$ ones, where $0\leq q_\lambda\leq n-2\lambda_1$. Specifically, we can use
\begin{equation}
\ket{\lambda,p_{\lambda0},q_\lambda}
=\Pi_{\lambda:p_{\lambda0}} \ket{x_{q_\lambda}},
\end{equation}
with
\begin{equation}
\ket{x_{q_\lambda}} \coloneqq \ket{01}^{\otimes \lambda_1} \otimes \ket0^{\otimes n-2\lambda_1-q_\lambda} \otimes \ket1^{\otimes q_\lambda}\;.
\end{equation}
Let us first evaluate
\begin{equation}
\label{eq:state_row_symmetrized}
\sum_{r \in \text{Row}(p_{\lambda0})}  R(r) \ket{x_{q_\lambda}} =
\Pi_{r\rightarrow c} \left(\ket{\Sigma_{\lambda_0}^{q_\lambda}}\otimes \ket1^{\otimes \lambda_1}\right),
\end{equation}
where $\ket{\Sigma_x^y}$ denotes the equal-weight superposition of all computation basis states on $x$ qubits with $x-y$ zeros and $y$ ones, which (up to normalization) is also known as \emph{Dicke state} on $x$ qubits \cite{dicke1954coherence,bartschi2019deterministic}. $\Pi_{r\rightarrow c}$ denotes the permutation of qubits needed to obtain the "column-standard" Young tableau $p_{\lambda0}$ from an analogous "row-standard" Young tableau where the numbers first increase in the row direction and then in column direction. In other words, if we think of the qubits being associated to the tiles of the Young diagram $\lambda$, then the qubits in the first row are in state $\ket{\Sigma_{\lambda_0}^{q_\lambda}}$, and the qubits in the second row are in state $\ket1^{\otimes \lambda_1}$.

Next, for a two-row standard Young tableau $p_{\lambda 0}$, we have
\begin{equation}
\sum_{c \in \text{Col}(p_{\lambda0})} \operatorname{sgn}(c) R(c)
= (\operatorname{id}_{\changevthree{2}}-\tau)^{\otimes \lambda_1}\otimes \operatorname{id}_2^{\otimes n-2\lambda_1} = (\ket{\Psi}\bra{\Psi})^{\otimes \lambda_1} \otimes \operatorname{id}_2^{\otimes n-2\lambda_1}\;,
\end{equation}
where $\ket\Psi$ is the 2-qubit singlet state
\begin{equation}
\ket{\Psi} = \frac{1}{\sqrt2}(\ket{01}-\ket{10})\;,
\end{equation}
and $\tau$ denotes the swap operator acting on two qubits. The qubits in the second row of $\lambda$ in the state of Eq.~\eqref{eq:state_row_symmetrized} are fixed to $\ket1$, so applying $\ket{\Psi}\bra{\Psi}$ to each of the first $\lambda_1$ columns has the same effect as applying $\ket{\Psi}\bra{\Psi} (\ket0\bra0\otimes \operatorname{id}_2)$. Applying $\ket0\bra0$ to the first $\lambda_1$ qubits of $\ket{\Sigma_{\lambda_0}^{q_\lambda}}$ yields $\ket0^{\otimes \lambda_1}\otimes \ket{\Sigma_{\lambda_0-\lambda_1}^{q_\lambda}}$. Thus, we find:
\begin{equation}
\ket{\lambda,p_{\lambda0},q_\lambda} = \Pi_{\lambda:p_{\lambda0}} \ket{x_{q_\lambda}}
= \ket{\Psi}^{\otimes \lambda_1} \otimes \ket{\Sigma_{n-2\lambda_1}^{q_\lambda}}.
\end{equation}

Now, we are ready to evaluate
\begin{equation}
\label{eq:combinatorial_sum}
\begin{multlined}
F^{\bm{i},\lambda}_{q_\lambda,q_\lambda'} \coloneqq \bra{\lambda, q_\lambda, p_{\lambda 0}} A_{\bm{i}} \ket{\lambda, q_\lambda', p_{\lambda 0}}\\=
    (\bra{\Psi}^{\otimes \lambda_1} \otimes \bra{\Sigma_{n-2\lambda_1}^{q_\lambda}})
    \left(\sum_{p_{\bm{i}} \in P_{\bm{i}}} p_{\bm{i}}\right)
    \left(\ket{\Psi}^{\otimes \lambda_1} \otimes \ket{\Sigma_{n-2\lambda_1}^{q_\lambda'}}\right)\\
    =
    \frac{1}{\sqrt{\binom{n-2\lambda_1}{q_\lambda}{\binom{n-2\lambda_1}{q_\lambda'}}}}
    \left(\sum_{s\in S^{q_\lambda}_{n-2\lambda_1}} \bra{\Psi}^{\otimes \lambda_1} \otimes \bra{s}\right)
    \left(\sum_{p_{\bm{i}} \in P_{\bm{i}}} p_{\bm{i}}\right)
    \left(\sum_{s'\in S^{q_\lambda'}_{n-2\lambda_1}} \ket{\Psi}^{\otimes \lambda_1} \otimes \ket{s'}\right)
    \;,
\end{multlined}
\end{equation}
where we used $S^x_y$ to denote the set of bitstrings of length $y$ with exactly $x$ ones.
This is a sum over (more than) exponentially many terms. Similarly to the previous Appendix, it can be evaluated efficiently by realizing that many summands have equal value. Thus, we instead sum over the different possible values multiplied with the number of summands with that value, which can be counted using combinatorics. Each summand is an overlap of two product states with a product operator in between. More precisely, we have a product of first $\lambda_1$ two-qubit overlaps, and then $n-2\lambda_1$ single-qubit overlaps.

For each summand in Eq.~\eqref{eq:combinatorial_sum}, let us denote by $L_{ab}$ with $a,b\in \{1,x,y,z\}$ the subset of two-qubit pairs:
\begin{equation}
L_{ab} \coloneqq \{(2l,2l+1): (p_{\bm{i}})_{2l}=\sigma_a, (p_{\bm{i}})_{2l+1}=\sigma_b, 0\leq l<\lambda_1\}\;,
\end{equation}
and let us write $f_{ab}=|L_{ab}|$ for the number of elements in those subsets.
The according overlap
\begin{equation}
\label{eq:2qubit_overlap}
    \bra\Psi (\sigma_a\otimes \sigma_b) \ket\Psi
\end{equation}
is $0$ if $a\neq b$, so we only need to consider subsets where $a= b$. The number of summands for given numbers $f_{aa}$ is the number of decompositions of the first $\lambda_1$ qubit pairs into the four subsets $L_{aa}$ with $a\in \{1,x,y,z\}$, which equals
\begin{equation}
\label{eq:first_qubits_counting}
    \frac{\lambda_1!}{f_{11}!f_{xx}!f_{yy}!f_{zz}!}\;.
\end{equation}
The value which the overlap on the first $\lambda_1$ qubit pairs contributes to each summand only depends on the numbers $f_{aa}$. The overlap in Eq.~\eqref{eq:2qubit_overlap} is given by $1$ if $a=b=1$, and $-1$ if $a=b$ otherwise. Thus, the overall contribution to each summand is
\begin{equation}
\label{eq:first_qubits_value}
    (-1)^{f_{xx}+f_{yy}+f_{zz}}\;.
\end{equation}

Next, let us consider the $n-2\lambda_1$ single-qubit overlaps. For each summand in Eq.~\eqref{eq:combinatorial_sum}, let us denote by $K_{iaj}$ for $i,j\in \{0,1\}$ and $a\in\{1,x,y,z\}$ the subset of the last $n-2\lambda_1$ qubits
\begin{equation}
K_{iaj}\coloneqq \{l: (p_{\bm{i}})_{2\lambda_1+l}=\sigma_a, s_l=i, s'_l=j, 0\leq l<n-2\lambda_1\}\;,
\end{equation}
and let us write $g_{iaj}=|K_{iaj}|$ for the number of elements in those subsets. The according overlap
\begin{equation}
\label{eq:single_qubit_overlap}
    \bra{i} \sigma_a \ket{j}
\end{equation}
is only non-zero if $i=j$ for $a\in \{1,z\}$ and $i\neq j$ for $a\in\{x,y\}$, so we can restrict to summands where only those 8 subsets are non-empty. The number of summands for given numbers $g_{iaj}$ is the number of decompositions of the set of the last $n-2\lambda_1$ qubits into the 8 subsets $K_{iaj}$, and is thus given by
\begin{equation}
\label{eq:second_qubits_counting}
    \frac{(n-2\lambda_1)!}{g_{010}!g_{111}!g_{0x1}!g_{1x0}!g_{0y1}!g_{1y0}!g_{0z0}!g_{1z1}!}\;.
\end{equation}

The contribution of the overlap on the last $n-2\lambda_1$ qubits to each summand only depends on the numbers $g_{iaj}$. The single-qubit overlap in Eq.~\eqref{eq:single_qubit_overlap} evaluates to $1$ for $g_{010}$, $g_{111}$, $g_{0X1}$, $g_{1X0}$ and $g_{0z0}$, $-1$ for $g_{1z1}$, $i$ for $g_{0y1}$, and $-i$ for $g_{1y0}$. Thus the overall contribution to each summand is
\begin{equation}
\label{eq:second_qubits_value}
    (-1)^{g_{1z1}} (-i)^{g_{0y1}} (i)^{g_{1y0}}\;.
\end{equation}
Overall, the number of summands for given $f_{aa}$ and $g_{iaj}$ is the product of Eq.~\eqref{eq:first_qubits_counting} and Eq.~\eqref{eq:second_qubits_counting}, and the value of each summand is given by the product of Eq.~\eqref{eq:first_qubits_value} and Eq.~\eqref{eq:second_qubits_value}. Plugging this into Eq.~\eqref{eq:combinatorial_sum} yields Eq.~\eqref{eq:sn_f_formula}. The constraints in Eq.~\eqref{eq:f_formula_constraints} are explained as follows. The first four constraints are due to the fact that the number of zeros and ones in $s$ and $s'$ is determined by $q_\lambda$ and $q_\lambda'$, respectively. The last four constraints correspond to the fact that the number of Pauli operators $\sigma_1$, $\sigma_x$, $\sigma_y$, $\sigma_z$ in $p_{\bm{i}}$ is given by $i_1$, $i_x$, $i_y$, and $i_z$, respectively.
\end{proof}

In a similar fashion to the previous Appendix, we can easily evaluate the runtime this method achieves in calculating all of the matrix elements. Note that each component is a sum over four independent variables due to the constraints, yielding a runtime of $\operatorname{O}\left(n^4\right)$. Taking into account the $\operatorname{O}\left(n^6\right)$ tensor components of $F$ yields the final runtime of $\operatorname{O}\left(n^{10}\right)$. Once again, it seems likely that the $\operatorname{O}\left(n^4\right)$ runtime for a single tensor component can be reduced to a smaller exponent. We will leave this open to further investigation.

\section{Ground state energy via structure coefficients}
In this appendix, we will give an alternative algorithm for finding the ground state energy only.
The runtime of this algorithm is slower than for the algorithm presented in the main text.
Yet it might help the reader develop further intuition for why the computation is tractable in polynomial time.
It is also simpler in that it does not require any knowledge of the Schur transform at all, but only determining the structure coefficients of the algebra $X$.

\subsection{Algorithm for general symmetry groups}
Let us start by phrasing the alternative ground-state energy algorithm for general symmetry groups and representations.

\begin{lemma}[Finding the ground state energy of symmetric Hamiltonians]
Consider a subalgebra $X$ of dimension $N$, and assume that the structure constants of $X$ in some preferred basis are known. Let $H\in A(X)$ be a Hamiltonian given in the preferred basis as in Eq.~\eqref{eq:hamiltonian_basis}. Then the ground state energy of $H$ can be found in time $\operatorname{O}\left(N^\omega\right)$.\label{lem:general_gse}
\end{lemma}
\begin{proof}
Consider the operator with indices:
\begin{equation}
\hat h^j_k\coloneqq \sum_i h_i X^{i,j}_k\;,
\end{equation}
which is nothing but the regular representation of $h$ for the algebra $X$. Then we have that their ground state energies are equal:
\begin{equation}
    \gse(H) = \gse(\hat h)\;,
\end{equation}
or, in tensor-network notation (c.f. Eq.~\eqref{eq:algebra_tensor}),
\begin{equation}
\begin{tikzpicture}
\atoms{square,lab={t=$\hat h$,p=-90:0.4}}{h/}
\draw (h-l)edge[mark={arr,f,s},ind=$l$]++(180:0.4) (h-r)edge[ind=$m$]++(0:0.4);
\end{tikzpicture}
\coloneqq
\begin{tikzpicture}
\atoms{square}{{h/lab={t=$h$,p=90:0.4},p={0,0.8}}}
\atoms{circ,lab={t=$X$,p=-90:0.4}}{a/}
\draw (a)edge[mark={arr,f,s}](h) (a-l)edge[mark={arr,f,s},ind=$l$]++(180:0.4) (a-r)edge[ind=$m$]++(0:0.4);
\end{tikzpicture}\;.
\end{equation}
This is because the regular representation is faithful, and the ground state energy of an operator is the same in any faithful representation.
Since $X$ has dimension $N$, the ground state energy of $\hat h$ can be found in time $\operatorname{O}(N^\omega)$.
\end{proof}

An advantage of this algorithm is that the only necessary information are the structure constants of $X$; no knowledge of the irreps of $X$ is needed. However, due to this we have poor scaling with the number of irreps $n_\lambda$, as the direct sum structure of $X$ is not necessarily known. Another disadvantage of this approach is that it only gives the ground state energy, rather than the ground state itself (in a representation that is not the regular representation).

\subsection{Algorithm for qubit permutation invariance}
\label{sec:gse_qubit_permutation}
Let us now apply the algorithm in Lemma~\ref{lem:general_gse} to the case of qubit permutation invariance.
\begin{corollary}
The ground state energy of a permutation-symmetric Hamiltonian on $n$ qubits, given as $h_i$ in the basis of symmetrized Pauli monomials above, can be computed in time $\operatorname{O}\left(n^{3\omega}\right)$ via Lemma~\ref{lem:general_gse}.
\end{corollary}
\begin{proof}
All that is needed for applying Lemma~\ref{lem:general_gse} are the structure constants of the algebra $X$, which are computed in the following.
Let us start by showning how these structure coefficients are computed in general, given only the representation $A$.
If we normalize $A$ such that it defines an orthonormal set of operators,
\begin{equation}
\label{eq:a_normalization}
\begin{tikzpicture}
\atoms{square}{{0/lab={t=$A$,p=-90:0.35}},{1/p={0,-1.2},lab={t=$A^*$,p=90:0.35}}}
\draw[nqubit,mark={arr,f,e},mark={arr,f,s},rc] (0)--++(-0.5,0)|-(1);
\draw[nqubit,rc] (0)--++(0.5,0)|-(1);
\draw (0)edge[ind=$a$]++(90:0.5) (1)edge[ind=$b$]++(-90:0.5);
\end{tikzpicture}
=
\begin{tikzpicture}
\draw[ind=$a$,mark={lab=$b$,b}] (0,0)--++(0,0.7);
\end{tikzpicture}
\;,
\end{equation}
we can calculate $X$ via
\begin{equation}
\label{eq:x_from_a}
\begin{tikzpicture}
\atoms{circ}{0/lab={t=$X$,p=-30:0.3}}
\draw (0)edge[mark={arr,f,s},ind=$a$]++(135:0.8) (0)edge[mark={arr,f,s},ind=$b$]++(45:0.8) (0)edge[ind=$c$]++(-90:0.8);
\end{tikzpicture}
=
\begin{tikzpicture}
\atoms{square}{{0/lab={t=$A$,p=-90:0.35}},{1/p={1.6,0},lab={t=$A$,p=-90:0.35}},{2/p={0.8,-0.8},lab={t=$A^*$,p=90:0.35}}}
\draw[nqubit,mark={arr,f,e},mark={arr,f,s},rc] (0)--++(-0.6,0)|-(2);
\draw[nqubit,mark={arr,f,e}] (0)--(1);
\draw[nqubit,rc] (1)--++(0.6,0)|-(2);
\draw (0)edge[ind=$a$]++(90:0.5) (1)edge[ind=$b$]++(90:0.5) (2)edge[ind=$c$]++(-90:0.5);
\end{tikzpicture}\;.
\end{equation}
Here, we use the fact that by definition of $X$, $A$ is a faithful definition of $X$.
Note that the $*$ denotes complex conjugation of the tensor. 
We cannot directly efficiently contract the tensor network on the right-hand side of Eq.~\eqref{eq:x_from_a}, since the thick indices have a bond dimension of $2^n$. This problem can be solved by realizing that $A$ can be written as an MPO, as we have seen in Eq.~\eqref{eq:a_mpo}.
Plugging this MPO representation into Eq.~\eqref{eq:x_from_a}, we obtain
\begin{equation}
\begin{tikzpicture}
\atoms{circ}{0/lab={t=$X$,p=-30:0.3}}
\draw (0)edge[triple,mark={arr,f,s},ind=$\vec i$]++(135:0.8) (0)edge[triple,mark={arr,f,s},ind=$\vec j$]++(45:0.8) (0)edge[triple,ind=$\vec k$]++(-90:0.8);
\end{tikzpicture}
=
\begin{tikzpicture}
\atoms{square,yscale=1.3,lab={t=$+$,p={0,0}}}{00/,10/p={0.7,0},20/p={2.1,0},01/p={0,0.7},11/p={0.7,0.7},21/p={2.1,0.7}}
\atoms{square,yscale=1.3,lab={t=$+$,p={0,0}},lab={t=$*$,p={0.13,0.32}}}{02/p={0,1.4},12/p={0.7,1.4},22/p={2.1,1.4}}
\atoms{circ,tiny,lab={t=$000$,p=180:0.4}}{0x/p={-0.5,0},0y/p={-0.5,0.7},0z/p={-0.5,1.4}}
\draw (20-r)edge[triple,ind=$\vec i$]++(0:0.4) (21-r)edge[triple,ind=$\vec j$]++(0:0.4) (22-r)edge[triple,ind=$\vec k$]++(0:0.4);
\draw (10-r)edge[triple,mark={three dots,a}]++(0:0.3) (11-r)edge[triple,mark={three dots,a}]++(0:0.3) (12-r)edge[triple,mark={three dots,a}]++(0:0.3);
\draw (20-l)edge[triple]++(180:0.3) (21-l)edge[triple]++(180:0.3) (22-l)edge[triple]++(180:0.3);
\draw[triple] (00-r)--(10-l) (01-r)--(11-l) (02-r)--(12-l);
\draw[triple] (0x)--(00-l) (0y)--(01-l) (0z)--(02-l);
\draw[rc] (00-t)--(01-b) (01-t)--(02-b) (02-t)--++(0,0.3)--++(0.3,0)|-($(00-b)+(0,-0.3)$)--(00-b);
\draw[rc] (10-t)--(11-b) (11-t)--(12-b) (12-t)--++(0,0.3)--++(0.3,0)|-($(10-b)+(0,-0.3)$)--(10-b);
\draw[rc] (20-t)--(21-b) (21-t)--(22-b) (22-t)--++(0,0.3)--++(0.3,0)|-($(20-b)+(0,-0.3)$)--(20-b);
\end{tikzpicture}\;,
\end{equation}
where we have vertically shrinked the $+$ tensors and the distance between the horizontal indices, and $*$ denotes complex conjugation. We see that now the tensor network can be contracted efficiently by contracting in the horizontal instead of the vertical direction. Indeed, the overall horizontal bond dimension is $n^9$, so contraction can be performed in polynomial runtime. To further speed up computation, we note that the values of each pair of horizontal indices on the left and right of a $+$ tensor differ by at most $1$. Thus, the tensor $+$ is constant-width band-diagonal in its horizontal indices, independent of the value of its vertical indices. Thus, the matrix corresponding to a column of three $+$ tensors is band-diagonal as well. Applying a band-diagonal matrix to a vector of bond dimension $d$ takes time $\operatorname{O}(d)$. Thus, contracting the above tensor network takes $n$ steps of runtime $\operatorname{O}(n^9)$, so runtime $\operatorname{O}(n^{10})$ in total.
\end{proof}

\subsection{Alternative combinatorial method for computing the structure coefficients of \texorpdfstring{$X$}{X}}
\label{app:x_structure_constants}
In this Appendix we give an alternative combinatorial algorithm for calculating structure constants of the algebra $X$, similar to the algorithm used in Appendix~\ref{sec:combinatorial_f}.
Again, the algorithm is slower than the tensor-network method in Appendix~\ref{sec:gse_qubit_permutation}, namely $\operatorname{O}(n^{15})$.
However, it might have the advantage of being more accessible to readers not familiar with tensor network.
\begin{lemma}
    The structure coefficients $X_{\bm k}^{\bm i, \bm j}$ of the completely symmetrized Pauli representation are given by:
    \begin{equation}
    \label{eq:x_structure_constants_explicit}
        X_{\bm k}^{\bm i, \bm j}=\sum\limits_{\substack{\{f_{ab}\}_{a,b\in\{1,x,y,z\}}\\/\text{Eq.\eqref{eq:f_constraint},Eq.\eqref{eq:k_constraint}}}}
            \frac{k_1!}{f_{11}!f_{xx}!f_{yy}!f_{zz}!}
            \frac{k_x!}{f_{1x}!f_{x1}!f_{yz}!f_{zy}!}
            \frac{k_y!}{f_{1y}!f_{y1}!f_{xz}!f_{zx}!}
            \frac{k_z!}{f_{1z}!f_{z1}!f_{xy}!f_{yx}!}
            (i)^{f_{xy}+f_{yz}+f_{zx}} (-i)^{f_{yx}+f_{xz}+f_{zy}}
    \end{equation}
    where the variables in the sum are non-negative integers subject to the constraints
    \begin{equation}
        \label{eq:f_constraint}
        \sum_{a\in\{1,x,y,z\}} f_{ab} = j_b\;,\qquad \sum_{b\in\{1,x,y,z\}} f_{ab} = i_a,
    \end{equation}
    and
    \begin{equation}
        \begin{aligned}
        f_{11}+f_{xx}+f_{yy}+f_{zz}&\eqqcolon k_1,\\
        f_{1x}+f_{x1}+f_{yz}+f_{zy}&=k_x,\\
        f_{1y}+f_{y1}+f_{xz}+f_{zx}&=k_y,\\
        f_{1z}+f_{z1}+f_{xy}+f_{yx}&=k_z.
        \end{aligned}
        \label{eq:k_constraint}
    \end{equation}
    \label{lemma:x_struct_consts}
\end{lemma}
\begin{proof}
For calculating the structure constants $X$, we first note that
\begin{equation}
A_{\bm j} =
    \frac{1}{i_1!i_x!i_y!i_z!} \sum_{\pi\in \mathrm{S}_n} R(\pi) \left(\sigma_1^{\otimes i_1} \otimes \sigma_x^{\otimes i_x} \otimes \sigma_y^{\otimes i_y} \otimes \sigma_z^{\otimes i_z}\right) R^{-1}(\pi)
    = \sum_{p_{\bm i} \in P_{\bm i}} p_{\bm i}\;,
\end{equation}
where $P_{\bm i}$ is the set of Pauli words with $i_a$ times $\sigma_a$ for $a\in\{1,x,y,z\}$. Now we evaluate the product:
\begin{equation}
\label{eq:x_coefficients_step1}
    A_{\bm i}\cdot A_{\bm j} 
    = \sum_{p_{\bm i} \in P_{\bm i}, p_{\bm j} \in P_{\bm j}} p_{\bm i} p_{\bm j}
    = \sum_{\bm k, p_{\bm k}\in P_{\bm k}}
    \sum_{\substack{p_{\bm i} \in P_{\bm i}, p_{\bm j} \in P_{\bm j}:\\ p_{\bm i} p_{\bm j} = \alpha_{p_{\bm i}, p_{\bm j}, p_{\bm k}}\cdot p_{\bm k}}} \alpha_{p_{\bm i}, p_{\bm j}, p_{\bm k}} \cdot p_{\bm k}.
\end{equation}
This is a sum over products of exponentionally many Pauli words. The idea to evaluate this is that many of the summands have equal value, so it suffices to sum over a few different values multiplied by the number of summands with that value.

    For every summand, define the subsets of qubits $L_{ab}$ for $a,b\in\{1,x,y,z\}$,
    \begin{equation}
    L_{ab}\coloneqq \{l: (p_{\bm i})_l=\sigma_a, (p_{\bm j})_l=\sigma_b, 0\leq l<n\},
    \end{equation}
    and let
    \begin{equation}
    f_{ab}\coloneqq |L_{ab}|
    \end{equation}
    be the numbers of elements in those subsets. Since every Pauli operator $i_l$ at a qubit $l$ is paired with some other Pauli operator $j_l$, $f_{ab}$ fulfill the constraints in Eq.~\eqref{eq:f_constraint}. The multiplication algebra of Pauli operators directly implies Eq.~\eqref{eq:k_constraint}.
    
    Let us now count how many Pauli words there are in the sum for a fixed set of numbers $f_{ab}$ and a fixed resulting Pauli word $p_{\bm k}$. Every triple $p_{\bm i}$, $p_{\bm j}$, $p_{\bm k}$ corresponds to a decomposition of each $\bm k_c$-element set of qubits $\{l:(p_{\bm k})_l=\sigma_c\}$ for $c\in\{1,x,y,z\}$ into four subsets $L_{ab}$ for the four different combinations $a,b\in\{1,x,y,z\}$ with $\sigma_a\sigma_b\propto \sigma_c$ under the Pauli algebra. For each $c$, the number of decompositions into the corresponding four subsets is given by
    \begin{equation}
        \frac{\bm k_c!}{\prod_{a,b: \sigma_a\sigma_b\propto \sigma_c} f_{ab}!}\;.
    \end{equation}
    In total, the number of decompositions into four subset\changevthree{s} for different $c$ is given by
    \begin{equation}
    \label{eq:x_coefficients_counting}
    \frac{k_1!}{f_{11}!f_{xx}!f_{yy}!f_{zz}!}
    \frac{k_x!}{f_{1x}!f_{x1}!f_{yz}!f_{zy}!}
    \frac{k_y!}{f_{1y}!f_{y1}!f_{xz}!f_{zx}!}
    \frac{k_z!}{f_{1z}!f_{z1}!f_{xy}!f_{yx}!}\changevthree{.}
    \end{equation}
    Finally, the prefactor $\alpha_{p_{\bm i}, p_{\bm j}, p_{\bm k}}$ in Eq.~\eqref{eq:x_coefficients_step1} only depends on the $f_{ab}$. Using the Pauli algebra,
    \begin{equation}
    \begin{aligned}
        \sigma_x\sigma_y&=i\sigma_z & \sigma_y\sigma_z&=i\sigma_x & \sigma_z\sigma_x&=i\sigma_y\\
        \sigma_y\sigma_x&=-i\sigma_z & \sigma_z\sigma_y&=-i\sigma_x & \sigma_x\sigma_z&=-i\sigma_y\;,
    \end{aligned}
    \end{equation}
    it is given by
    \begin{equation}
    \label{eq:x_coefficients_value}
    \alpha_{p_{\bm i}, p_{\bm j}, p_{\bm k}} = (i)^{f_{xy}+f_{yz}+f_{zx}} (-i)^{f_{yx}+f_{xz}+f_{zy}}\;.
    \end{equation}
    Using Eq.~\eqref{eq:x_coefficients_counting} and Eq.~\eqref{eq:x_coefficients_value} in Eq.~\eqref{eq:x_coefficients_step1} directly yields Eq.~\eqref{eq:x_structure_constants_explicit}.
\end{proof}

Let us quickly discuss the complexity of the computation of $X_{\bm k}^{\bm i, \bm j}$. In the summation of Eq.~\eqref{eq:x_structure_constants_explicit}, we sum over $16$ variables within a range of the order $n$, so if we naively evaluate the sum, we already obtain a polynomial runtime $\operatorname{O}(n^{16})$. However, due to the constraint Eq.~\eqref{eq:f_constraint}, we can reduce the summation to only $9$ variables $f_{ab}$ with $a,b\in\{x,y,z\}$. Eq.~\eqref{eq:k_constraint} poses another three independent constraints, reducing the summation to $6$ variables. Thus, an individual entry $X_{\bm k}^{\bm i, \bm j}$ can be calculated in $\operatorname{O}(n^6)$ runtime, whereas all $\operatorname{O}(n^9)$ coefficients together take runtime $\operatorname{O}(n^{15})$.

\section{Permutation invariance for qudits}
\label{sec:qudits}
In this appendix, we sketch the generalization of the proposed algorithms to permutation invariant systems with $n$ $d$-dimensional qudits instead of qubits.
In this case, the dimension of the symmetrized basis is $\operatorname{dim}(X) = \binom{n + d^2 - 1}{d^2-1} = \operatorname O\left(n^{d^2-1}\right)$. This can be derived by utilizing a stars and bars counting argument, representing the number of ways to place $d^2-1$ bars among $n$ stars. Let us now determine the numbers $n_\lambda$ and $n_q$ that determine the runtimes of our algorithms, and can also be used to rederive $\operatorname{dim}(X)=\operatorname{O}(n_\lambda n_q^2)$.
The irreducible representations of $SU(d)$ are identified with Young diagrams with $d-1$ rows, which can be labelled by a sequence $\mathbf j=\{j_i\}_{0<i<d-1}$ with $j_i\geq j_{i+1}$ where $j_i$ is the number of boxes in the $i$th row.
The representation acting on a single qudit is the $d$-dimensional fundamental irrep, corresponding to a Young diagram with only one box, which we denote by $\mathbf d$.
The overall representation is $R=\mathbf d^{\otimes n}$
In order to decompose it into irreps, we only need to know the tensor product of an arbitrary irrep $\mathbf j$ with the fundamental irrep $\mathbf d$.
It is given by a direct sum of at most $d+1$ irreps with at most $x+1$ boxes, if $x$ is the number of boxes of $\mathbf j$.
Specifically, the direct sum consists of irreps which are obtained by adding a box to any row, as well as an irrep obtained by removing a box from every row (as long as these operation still yield valid Young diagrams).
So we see that for $n$ qubits, the irreps in the decomposition of $R$ are Young diagrams with at most $n$ boxes.
The number of these Young diagrams is $\operatorname{O}(n^{d-1})$.
The maximum dimension of an irrep is $\operatorname{O}(n^{\frac{d(d-1)}2})$, as can be seen from the asymptotics of the formula
\begin{equation}
    |\lambda| = \left[\prod_{1\leq i < j \leq d} (\lambda_j - \lambda_i +j -i)\right]/ \left[ \prod_{m=1}^d m! \right] 
\end{equation}
given in Ref.~\cite{harrowschur}, with $\lambda_k \approx n(k/d)$.

Let us now look at how the computation of the matrix elements $F$ via Eq.~\eqref{eq:fmatrix_computation} carries over to the case of qudits.
Analogous to Eq.~\eqref{eq:rep_mpo_representation}, the symmetrized Pauli basis for qudits can be represented as an MPO of bond dimension $\operatorname O\left(n^{d^2-1}\right)$.
This MPO is still band-diagonal.
Next, the tensors in the MPO representation in Eq.~\eqref{eq:rep_mpo_representation} are the Clebsch-Gordan coefficients $C_{(\lambda, q_\lambda), (\mathbf d, q_{\mathbf d})}^{(\lambda', q_\lambda')}$, which can be found in Section~18.2.6.\ of Ref.~\cite{Vilenkin1992}.
For an appropriate choice of basis, $C$ band diagonal.
That is, for a fixed $\lambda$ and $q_\lambda$, there is a constant number of $\lambda'$ and $q_\lambda'$ with $C_{(\lambda, q_\lambda), (\mathbf d, q_{\mathbf d})}^{(\lambda', q_\lambda')}\neq 0$.
For the $\lambda$ part, this follows from the fact that tensor product with the fundamental representation  yields at most $d+1$ irrep factors as discussed above.
Overall, the total bond dimension in Eq.~\eqref{eq:fmatrix_computation} is $\operatorname{O}(n_q^2 n^{d^2-1}) = \operatorname{O}(n^{2d^2-d-1})$.
Since the overall matrix (consisting of twice $C$ and the tensor $+$) is band diagonal, a single vector-matrix multiplication in the contracting of Eq.~\eqref{eq:fmatrix_computation} from left to right takes time $\operatorname{O}(n^{2d^2-d-1})$.
Since there are $\operatorname{O}(n_\lambda) = \operatorname{O}(n^{d-1})$ involved irreps, and the contraction consists of $n$ vector-matrix multiplications, the total runtime is $\operatorname{O}(n^{2d^2-1})$.
So we see that we do get a polynomial runtime also for qudits, but the exponents explode quickly with increasing local dimension, for example we have $\operatorname{O}(n^{17})$ for qutrits.

\end{document}